\documentclass{vldb}
\usepackage{graphicx}
\usepackage{balance}  
\usepackage{amssymb}
\usepackage{multirow}
\usepackage{amsmath}
\usepackage{color}
\usepackage{subfigure}
\usepackage{empheq}
\usepackage{algorithm}
\usepackage{algorithmic}
\usepackage{enumitem}
\usepackage{amssymb}
\usepackage{amsfonts}
\usepackage{url}
\usepackage{epstopdf}

\newtheorem{definition}{Definition}

\newtheorem{theorem}{Theorem}[section]

\newtheorem{corollary}[theorem]{Corollary}

\newcommand{\HPartition}{H-Partition}
\newcommand{\BPartition}{B-Partition}

\newcommand{\nop}[1]{}

\begin{document}


\pagenumbering{arabic}

\title{Memory Efficient De Bruijn Graph Construction}



%
%
%
%

\author{
{Yang Li, Pegah  Kamousi, Fangqiu Han, Shengqi  Yang, Xifeng Yan, Subhash  Suri}%
\vspace{1.5mm}\\
University of California, Santa Barbara\\[0.5mm]
\{yangli, pegah, fhan, sqyang, xyan, suri\}@cs.ucsb.edu
}

\maketitle

\begin{abstract}
Massively parallel DNA sequencing technologies are revolutionizing genomics research. Billions of short reads generated at low costs can be assembled for reconstructing the whole genomes. Unfortunately, the large memory footprint of the existing de novo assembly algorithms makes it challenging to get the assembly done for higher eukaryotes like mammals.  In this work, we investigate the memory issue of constructing de Bruijn graph, a core task in leading assembly algorithms, which often consumes several hundreds of gigabytes memory for large genomes. We propose a disk-based partition method, called Minimum Substring Partitioning (MSP), to complete the task using less than 10 gigabytes memory, without runtime slowdown.  MSP breaks the short reads into multiple small disjoint partitions so that each partition can be loaded into memory, processed individually and later merged with others to form a de Bruijn graph.  By leveraging the overlaps among the k-mers (substring of length k), MSP achieves astonishing compression ratio: The total size of partitions is reduced from $\Theta(kn)$ to $\Theta(n)$, where $n$ is the size of the short read database, and $k$ is the length of a $k$-mer.  Experimental results show that our method can build de Bruijn graphs using a commodity computer for any large-volume sequence dataset.

\medskip
\noindent \textbf{Source codes and datasets}: \url{grafia.cs.ucsb.edu/msp}

\end{abstract}

\section{Introduction}
High-quality genome sequencing is foundational to many critical biological and medical problems.  Recently, massively parallel DNA sequencing technologies \cite{mardis2008next}, such as Illumina  \cite{illumina} and SOLiD \cite{solid}, have been reducing the cost  significantly. The price for Human Whole Genome Sequencing at a $30X$ coverage has dropped to $\$3,750$ (www.knome.com).  The massive amount of short reads (short sequences with symbols $A, C, G, T$) generated by these next-generation techniques \cite{mardis2008next} quickly dominate the scene.  How to manage and process the Big Sequence Data becomes a database issue.

A key problem in genome sequencing is assembling massive short reads that are extracted from DNA segments.  The number of short reads can easily reach one billion; and the length of each read varies from a few tens of bases to several hundreds. Figure \ref{fig:assembly} shows a sequence assembly process, where three short sequences are assembled to a longer sequence based on their overlaps.

\begin{figure}[h]
\centering
\includegraphics[width=0.3\textwidth]{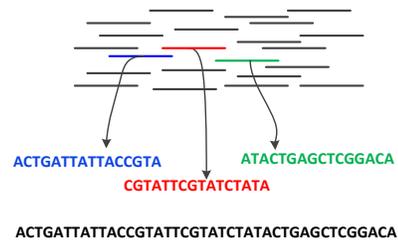}
\caption{Sequence Assembly}
\label{fig:assembly}
\end{figure}

The above process, called De novo assembly, has been extensively studied in the past decade.  There are two kinds of approaches: the overlap-layout-consensus approach \cite{myers2000whole, platt2009forge}, and the de Bruijn graph approach \cite{pevzner2001eulerian, zerbino2008velvet, simpson2009abyss, butler2008allpaths, li2010novo}. The overlap-layout-consensus approach builds an overlap graph between short reads. Due to the sheer size of the overlap graph (each read can overlap with many other reads), this approach is more suitable for small genomes.  The de Bruijn graph approach breaks short reads to k-mers (substring of length k) and then connects k-mers according to their overlap relations in short reads.  It can assemble larger quantities (e.g., billions) of short reads with greater coverage.

Despite their popularity, large memory consumption is a bottleneck for both approaches \cite{miller2010assembly}.  For the short read sequences generated from mammalian-sized genome, algorithms such as Euler \cite{pevzner2001eulerian}, Velvet \cite{zerbino2008velvet}, AllPaths \cite{butler2008allpaths} and SOAPdenovo  \cite{li2010novo} have to consume hundreds of gigabytes memory.  Figure \ref{fig:eachstep} shows a breakdown of memory and runtime consumption in SOAPdenovo  \cite{li2010novo} on a 258.7 GB Cladonema short read dataset and a 137.5 GB Lake Malawi cichlid (fish) short read dataset\footnote{A de Bruijn graph based assembly process consists of six steps: error correction (optional), de Bruijn graph construction, contig generation, reads remapping, scaffolding (optional) and gap closure (optional). The last two steps are applicable when pair end information is available.}. Obviously, the most memory consuming and time intensive part is the de Bruijn graph construction step.  Similar results were also reported for other datasets \cite{li2010novo}.  In this work, we resort to a novel disk-based approach to tackle this bottleneck, using less than 10 gigabytes memory, without runtime slowdown.

\begin{figure}[h]
\centering
\subfigure[\small\textit{Peak Memory}]{\label{fig:mem_step}\includegraphics[width=0.235\textwidth]{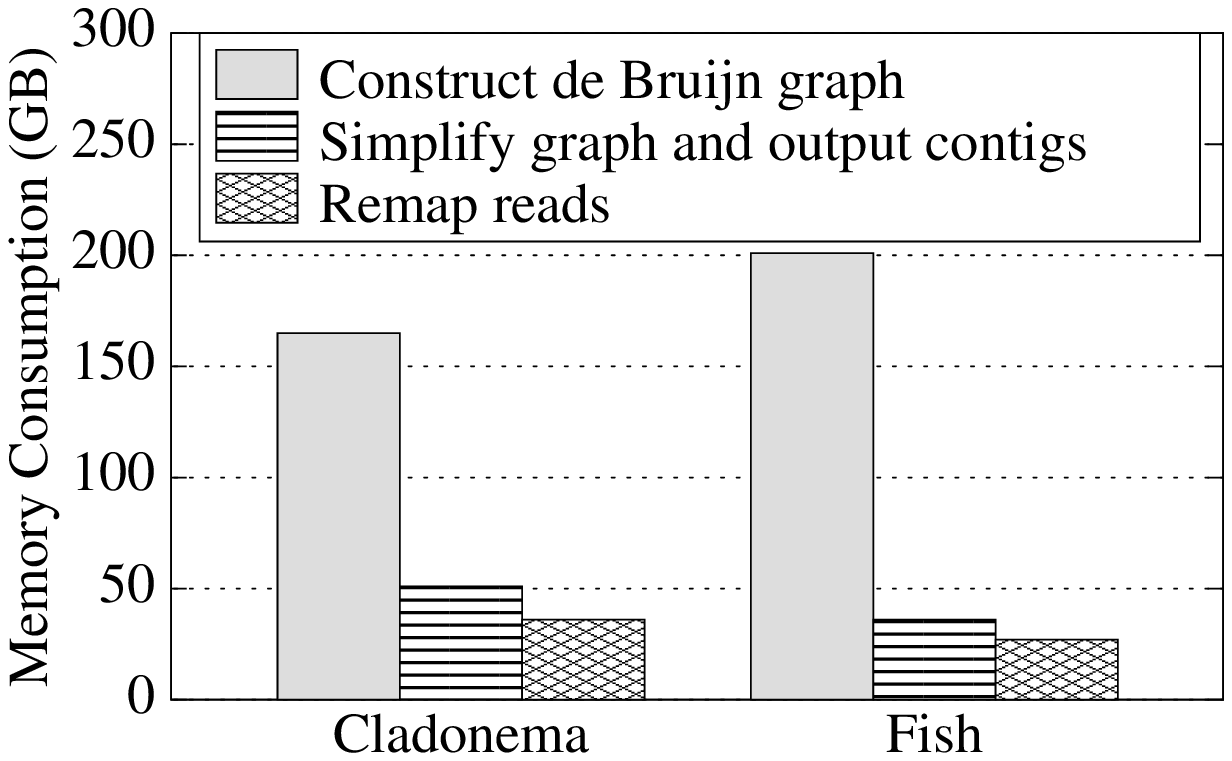}}
\subfigure[\small\textit{Running Time}]{\label{fig:time_step}\includegraphics[width=0.235\textwidth]{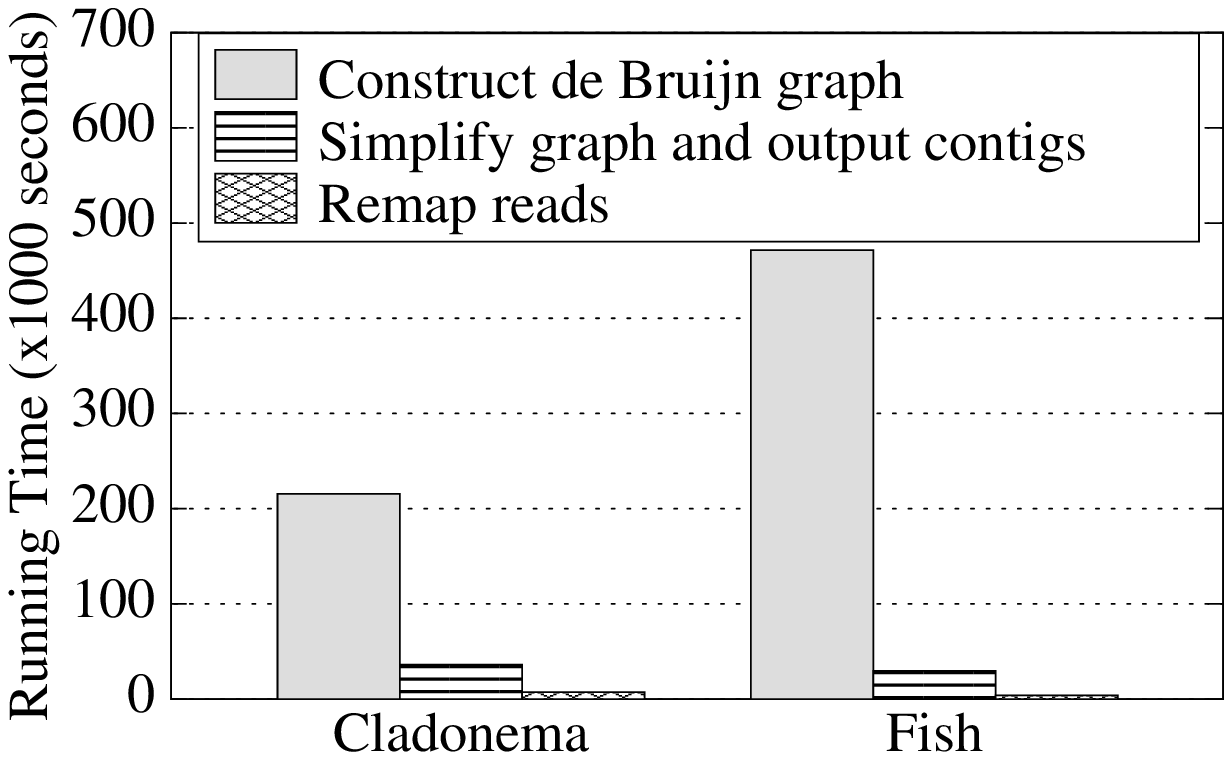}}
\caption{SOAPdenovo: Statistics of Computational Complexity at Each Assembly Step}
\label{fig:eachstep}
\end{figure}

In a de Bruijn graph, each vertex represents a k-mer.  In order to build the graph, we have to identify the same k-mers scattered in different short reads. A straightforward solution is to build a hash table. We can encode each symbol, A, C, G, and T using 2 bits.  In the aforementioned 137.5 GB fish datasets (the read length is 101), when $k = 59$, there are about 11.8 billion distinct k-mers including reverse complements. Assuming a load factor of $2/3$ for the hash table, we could expect the hash table to take nearly 283 GB memory, which is too large.

Alternatively, one can apply a disk-based partition-merge approach, which is popular in databases.  Given a set of short reads $S$, there are two classic scatter-gather methods to identify duplicate k-mers: (1) partition $S$ horizontally into disjoint subsets, $S_1, S_2, \ldots, S_t$, for each subset $S_i$, generate a hash table $H_i$ of their k-mers in main memory, output a sorted copy $H_i$ to disk, and then  merge $H_1, H_2, \ldots, H_t$; (2) partition all k-mers from $S$ into disjoint subsets $S_1, S_2, \ldots, S_t$ based on their last few symbols, for each subset $S_i$, create a hash table $H_i$, build a k-mer mapping and output  $H_i$ to disk, and then combine them. Both methods do not require a large amount of memory; but they are slow.  The first solution, requiring multiple disk scans and sorts, is hopeless.  The second one has to generate a huge number of k-mers in the first step.  For the 258.7 GB Cladonema dataset, with k = 59, the disk file of k-mers is close to 3TB and the time used to finish duplicate mapping is around 30 hours.

In this paper, we re-examine the second scatter-gather approach and find a drawback existing in its k-mer partitioning strategy.  Many k-mers generated from the same short read, though having large overlaps inside, are distributed to different partitions, which caused huge overhead.  Inspired by this discovery, we introduce a new concept, called \emph{minimum substring partitioning} (MSP).  MSP breaks short reads to pieces larger than k-mers; each piece contains k-mers sharing a common minimum substring with fixed length $p$, $p\leq k$.  The effect is equivalent to compressing consecutive k-mers using the original sequences.  We demonstrate that this compression approach does not introduce significant computational overhead, but could lead to 10-15 times smaller partitions, thus improving performance dramatically. It is observed that the size of MSP partitions is only slightly larger than the original sequences.  Based on a random string model, we analytically derive the expected size of minimum substring based partitions, which is reduced from $\Theta(kn)$ to $\Theta(n)$, where $n$ is the size of the short read database, and $k$ is the length of a $k$-mer.  Furthermore, we prove that the size of the largest partitions decreases exponentially with respect to $p$, indicating that it is very memory-efficient. When $p=12$,  the memory consumption is less than 10G for all the real datasets we tested.

Our main contribution is the development of an innovative disk-based partitioning strategy for solving a critical graph construction problem in genome sequence assembly.  Our solution is disk-based, using a small amount of memory without runtime performance loss. To the best of our knowledge, our study is the first work that introduces minimum substring partitioning, studies its properties, and successfully applies it to de novo sequence assembly, a critical problem in genome analysis.  Experimental results show that our method can build de Bruijn graphs using a commodity computer for any large-volume sequence dataset.

\section{Preliminaries}

\begin{definition}[Short Read, K-Mer]
A short read is a string over alphabet $\Sigma$.   A $k$-mer is a string whose length is $k$. Given a short read $s$, $s[i,j]$ denotes the substring of $s$ between the $i_{th}$ and $j_{th}$ (both inclusive) elements.  $s$ can be broken into $m-k+1$ $k$-mers, written as $s[1, k]$, $s[2, k+1], \ldots$, $s[m-k+1, m]$. K-mers  $s[i, k+i-1]$, $s[i+1, k+i]$ are called adjacent in $s$.
\end{definition}

For a short read $s$, we can view k-mers generated in a way that a window with width $k$ slides through $s$. Two k-mers, $\alpha$ and $\beta$, are adjacent from $\alpha$ to $\beta$ if and only if the last $k-1$ substring of $\alpha$ is the first $k-1$ substring of $\beta$.  Let $S$ be a short read set $S=\{s_i\}$.  A k-mer extracted from $s_i$, $s_i[j, j+k-1]$, is written as $s_{i,j}$.

\begin{definition}[De Bruijn Graph]
Given a short read set $S=\{s_i\}$, a de Bruijn graph $G=\{V, E\}$ is constructed by creating a vertex for every distinct k-mer in $S$ and connecting two vertices with a directed edge if their corresponding k-mers are adjacent in at least one short read.
\end{definition}

\begin{figure}[h]
\centering
\epsfig{file=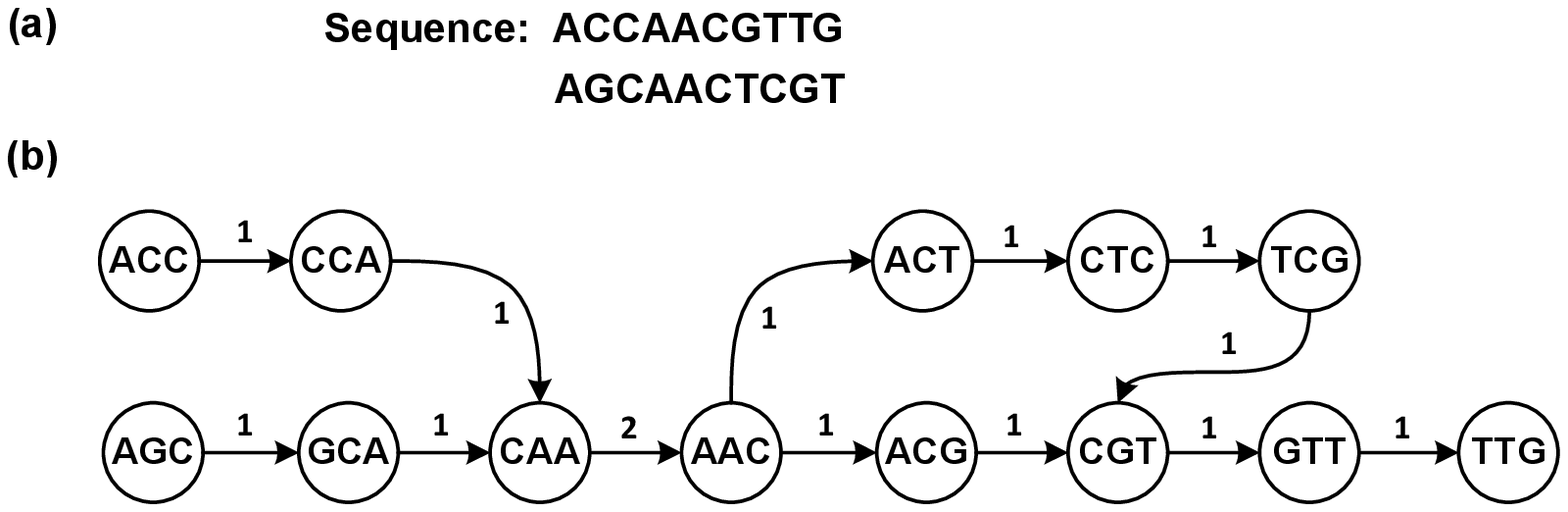, width=0.47\textwidth}
\caption{A de Bruijn Graph Example: k=3}
\vspace{-3mm}
\label{fig:debruijn}
\end{figure}

Figure \ref{fig:debruijn} shows a de Bruijn graph generated from two short reads with $k$ being $3$.  The edge weight shows the number of times the two adjacent k-mers appear in short reads. For sake of simplicity, we do not depict the k-mers generated by the reverse complements of short reads (see details in Section \ref{sec:reverse}).

\subsection {K-mer Mapping}
Given a short read dataset, in order to build a de Bruijn graph, one has to map all the duplicate k-mers derived from different short reads into the same vertex.  If vertices are assigned with integer id's, e.g., starting at $1$, this is equivalent to mapping duplicate k-mers to the same id.  This process is called \emph{K-mer Mapping}. Once the mapping is built, by scanning the short reads, we can create the edge set for the de Bruijn graph naturally. Therefore, the task of building a de Bruijn graph is narrowed down to k-mer mapping and edge sequence generation.

%

\subsection{Scatter/Gather}
One solution to the memory bottleneck issue is to chunk data to several partitions and process them separately \cite{Teuhola93externalduplicate, teuhola1991minimal}.  In this section, we discuss two scatter/gather approaches derived from duplicate detection techniques and then show their space complexity.  The first solution is called Horizontal Partition (\HPartition).

\begin{enumerate}
\item Divide short read dataset $S$ to disjoint partitions with equal size, $S_1$, $S_2$, $\ldots$, $S_t$, such that each partition can be loaded into memory.

\item For each partition $S_i$, insert $k$-mers into a hash table $H_i$. Based on the insertion order, assign an increasing integer id, starting at $1$, to each distinct $k$-mer. Let $M_i$ be the k-mer mapping function in $S_i$.  $M_i$ is local.

\item For each partition $S_i$, output all k-mers $s_{i, j}$ (in this case, we need to output the k-mer itself and its index, $(i, j)$) together with the assigned id, in increasing order of $(i, j)$.  Let $P_i$ be the output sequences.

\item Merge $\{P_i\}$ to generate a global mapping function $M$ such that it satisfies the following constraint.  For any k-mer $\gamma$ extracted from partition $S_j$, let $S_i$ be the partition with the smallest $i$ that contains $\gamma$, then $M(\gamma)=M_i(\gamma)$.

\end{enumerate}

The output size of Step 3 is $\Theta(kn)$, where $n$ is the size of the short read database, and $k$ is the k-mer's length. Step 4 in \HPartition\ is costly.  It needs a sort/merge process to identify the duplicate k-mers in different partitions.

The main issue of \HPartition\ arises from the fact that the multiple occurrences of the same k-mer are not located in the same partition.  To overcome this issue, one common strategy is to do bucket partitioning.  Let $H$ be a hash function of k-mer.  We can generate $t$ partitions by distributing k-mer $s_{i,j}$ to the $H(s_{i,j})\mod t$ partition.  We can also use k-mers' last several symbols to scatter them into different partitions. This classic approach is called Bucket Partition (\BPartition).

\begin{enumerate}
\item Extract all k-mers from $S$ and put them to disjoint partitions, $S_1, S_2, \ldots, S_t$, according to $H(s_{i,j}) \mod t$.

\item For each partition $S_i$, insert $k$-mers into a hash table $H_i$ and assign an increasing integer id, starting at $\Sigma_{j=1}^{i-1} |S_j|$, to each distinct $k$-mer based on the insertion order, where $|S_j|$ is the number of distinct $k$-mers in partition $S_j$.  Let $M$ be this k-mer mapping function.  It is clear that $M$ is a global mapping function: each distinct k-mer in $S$ will have one unique id.

\item For each partition $S_i$, output all k-mers $s_{i, j}$ (in this case, we only need to output the index $(i,j)$, not the k-mer string) together with the assigned id, in increasing order of $(i, j)$.  Let $P_i$ be the output sequences.

\item Merge $\{P_i\}$ in increasing order of $(i, j)$.

\end{enumerate}

While Step 4 in \BPartition\ is much faster than that in \HPartition, the total size of all the partitions is the same $\Theta(kn)$, which could easily reach multiple terabytes for a large genome. In the following discussion, we introduce a new partitioning concept, minimum substring partitioning (MSP), that reduces the partition size to $\Theta(n)$.

\section{Minimum Substring Partitioning}
\label{sec:msp}
Bucket partitioning has high overhead since adjacent k-mers are likely distributed to different partitions, unless $H(s_{i,j}) \mod t$ $=$ $H(s_{i,j+1}) \mod t$. Karp and Rabin \cite{karp1987efficient} proposed a rolling hash function with the property that the hash value of consecutive k-mers can be calculated quickly.  However, it is unknown whether there exists such a hash function that with high probability, two adjacent k-mers could be mapped to the same partition.  In this study, we resort to another approach to bypass this problem.

\begin{definition} [Minimum Substring\cite{Robe+04}]
Given a string $s$, a length-p substring $r$ of $s$ is called the minimum p-substring (or pivot substring) of $s$, if $\forall s'$, $s'$ is a length-p substring of $s$, s.t., $r\leq s'$ ($\leq $ defined by lexicographical order). $s$ is said to be covered by $r$. The minimum p-substring of $s$ is written as $min_{p}(s)$.
\end{definition}

\begin{figure}[h]
\centering
\epsfig{file=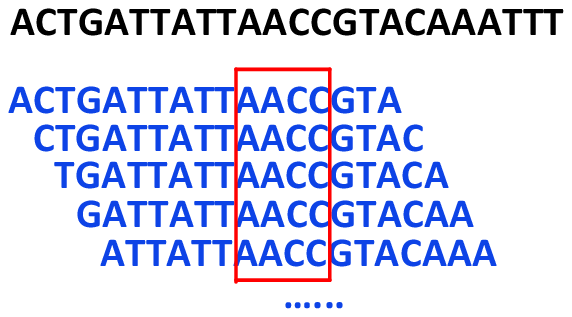, width=0.25\textwidth}
\caption{Minimum Substring Partitioning}
\label{fig:msp}
\end{figure}

Since two adjacent k-mers overlap with length $k-1$ substring, the chance for them to have the same minimum $p$-substring ($p < k$) could be very high. Figure \ref{fig:msp} illustrates that the first $5$ k-mers have the same minimum $4$-substring, $AACC$.  In this case, instead of generating these $5$ k-mers separately, one can just compress them using the original short read, to $ACTGATTATTAACCGTACAAA$, and output it to the partition corresponding to the minimum $4$-substring $AACC$. Formally speaking, given a short read $s=s_1 s_2 \ldots s_m$, if the adjacent $j$ k-mers from $s[i, i+k-1]$ to $s[i+j-1, i+j+k-2]$ share the same minimum $p$-substring $r$, then one can just output substring $s_i s_{i+1} \ldots s_{i+j+k-2}$ to partition $H(r) \mod t$ without breaking it to $j$ k-mers.  If $j$ is large, this compression strategy will dramatically reduce the partition size and runtime.

\begin{definition}[Minimum Substring Partitioning]
\label{def:msp}
Given a string $s=s_1 s_2 \ldots s_m$, $p\leq k \leq m$, minimum substring partitioning breaks $s$ to substrings with maximum length $\{s[i, j]| i+k-1\leq j, 1\leq i, j\leq m\}$, s.t., all k-mers in $s[i,j]$ share the same minimum p-substring.  $s[i,j]$ is also called super k-mer.
\end{definition}

According to minimum substring partitioning, larger $p$ will likely break a sequence to several segments with different minimum $p$-substrings, thus increasing the total partition size.  On the other hand, a smaller $p$ will produce larger partitions that might not fit in the main memory. 

\begin{figure}[h]
\centering
\epsfig{file=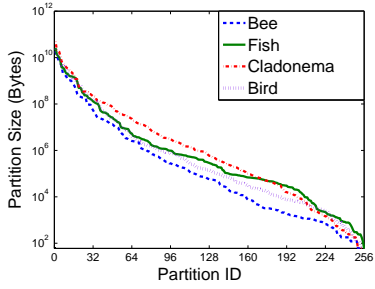, width=0.3\textwidth}
\vspace{-2mm}
\caption{Partition Size Distribution}
\vspace{-2mm}
\label{fig:distribution}
\end{figure}
Figure \ref{fig:distribution} shows the distribution of partition size with $p=4$ on the bee, fish, cladonema and bird datasets (ref. to Table 1 for details). The partitions are sorted according to their sizes. There are several large dominating partitions.  The value of $p$ determines the total size of partitions and the expected size of the largest partitions. In the following discussion, using a random string model, we prove that the expected total partition size is $\Theta (n)$, far smaller than $\Theta (kn)$ in \HPartition\ and \BPartition.  We will further show the lower and upper bound of the largest partition in MSP, which decreases exponentially with respect to $p$, indicating that MSP is very memory-efficient.

\subsection{Total Partition Size}

Let $l$ be the average number of breaks that MSP introduces in a given sequence dataset. That is, on average, MSP adds $l$ breaks to a sequence and divides it into multiple substrings $s[i_1, j_1]$, $s[i_2, j_2]$, $\ldots$, $s[i_{l+1}, j_{l+1}]$.  Let $m$ be the length of individual short reads. Suppose there are $n/m$ short reads, i.e., $n$ is the dataset size. We have the following theorem.

\begin{theorem}
\label{thrm:partitiontotalsize}
The total partition size is $\Theta (\frac{lk}{m}n+n)$.
\end{theorem}
\begin{proof}
Each break introduces a substring that overlaps its previous substring with $k-1$ symbols. We have $\frac{n}{m}l$ breaks.  Hence, the total partition size is $\Theta (\frac{lk}{m}n+n)$.
\end{proof}

\newcommand{\hs}{\hspace{-0.05cm}}

Assume a random string model with four symbols $0$, $1$, $2$, and $3$, each having equal probability to occur.  We first use a simulation method to demonstrate the average number of breaks for $1M$ short reads with length $m=100$.

\begin{figure}[h]
\centering
\subfigure[\small\textit{$p$ and $k$ }]{\label{fig:estimation-breaks-pk}\includegraphics[width=0.23\textwidth]{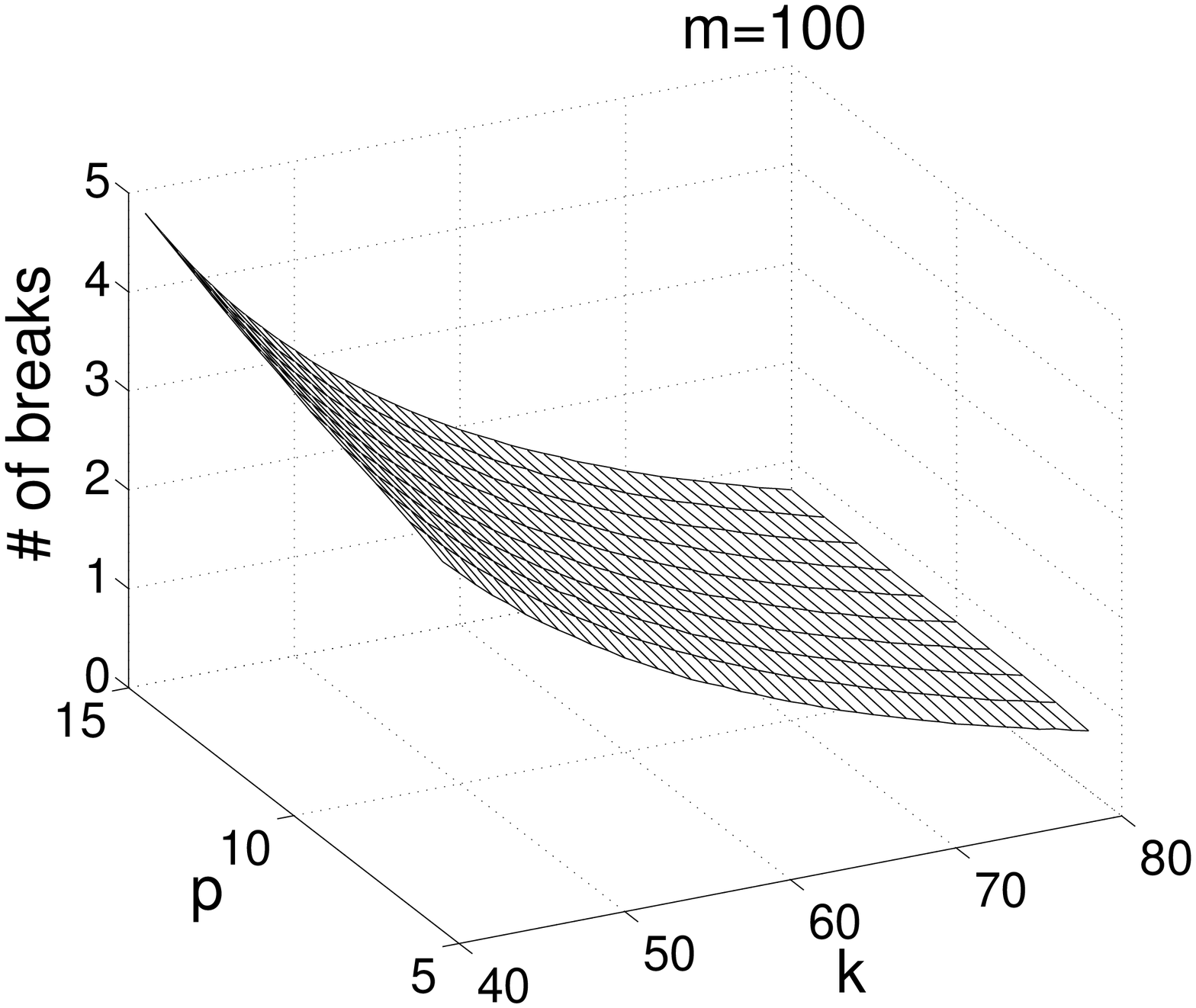}}
\subfigure[\small\textit{$m$}]{\label{fig:estimation-breaks-m}\includegraphics[width=0.23\textwidth]{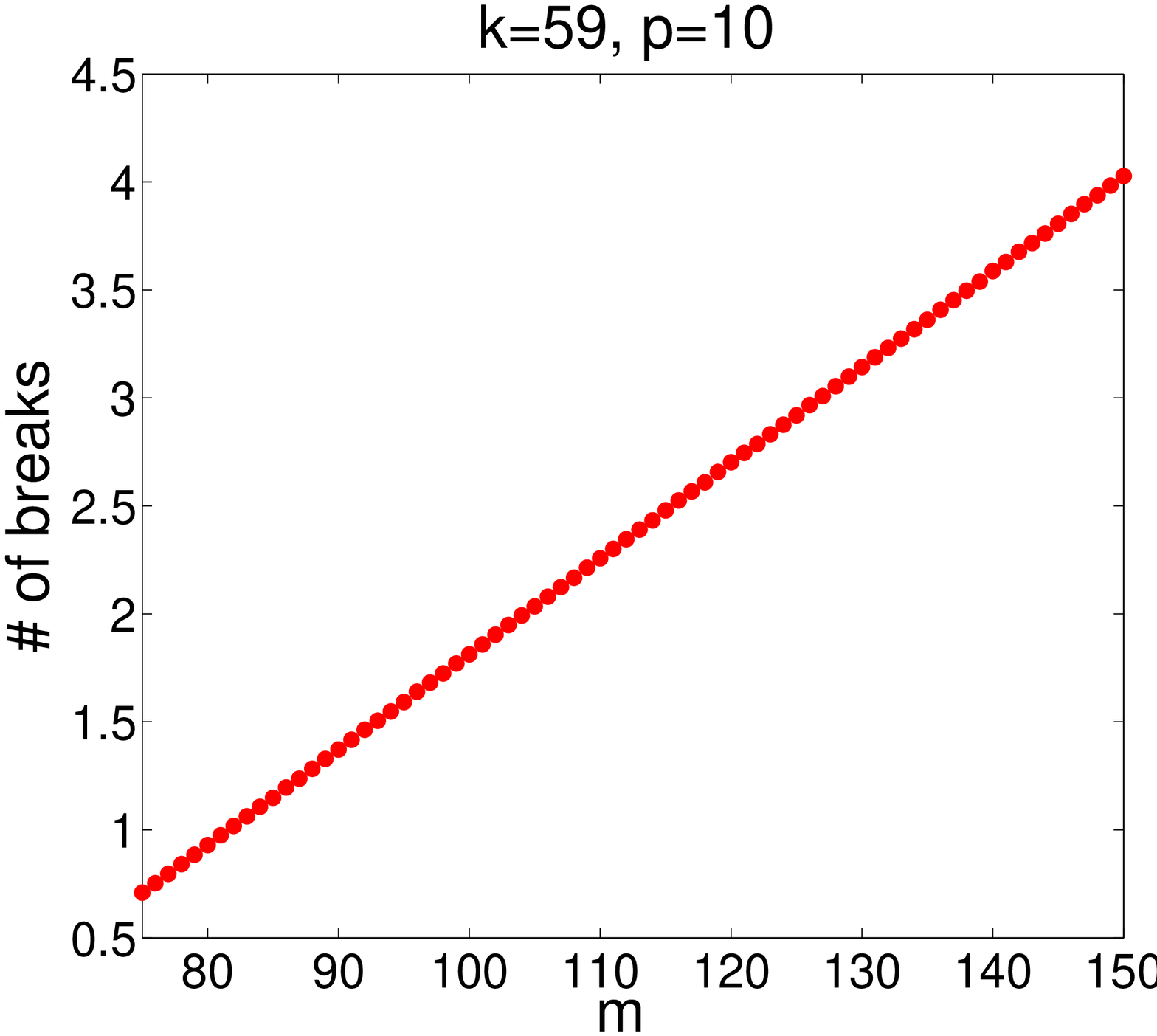}}
\caption{Average Number of Breaks}
\label{fig:estimation-breaks}
\end{figure}

Figure \ref{fig:estimation-breaks-pk} shows the expected number of breaks with respect to different $p$ and $k$ values.  When $p$ increases, the number of breaks increases.  When $k$ increases,  the number of breaks decreases.  Figure \ref{fig:estimation-breaks-m} shows the expected breaks of short reads with respect to different $m$ values, with $p=10$ and $k=59$.  It is observed that the average number of breaks increases proportionally with respect to $m$. We prove this in the following theorem.

\begin{theorem}\label{lp1m}
Let $l(m,k,p)$ be the average number of breaks under minimum substring partitioning. In a random string model, $l(m,k,p)\propto (m-k)$.
\end{theorem}
\begin{proof}
It is trivial to have $l=0$ when $m=k$, because the whole string has no break in this situation. Consider the difference between $l(m, k, p)$ and $l(m-1, k, p)$. In an $m$ length string, let $P_1(k,p)$=Pr$\{$the minimum $p$-$substring$ of the last $k$-$mer$ is different from the second last one$\}$. This equals to $P_1(k,p)$=Pr$\{$the first or the last p-substring is the only smallest p-substring$\}$. Since $P_1$ is only related to the last $k+1$ characters, it is not related to $m$. Then we have,
\begin{eqnarray*}
l(m, k, p)=l(m-1, k, p)+P_1(k,p)\\
= \cdots =P_1(k,p)\cdot (m-k).
\end{eqnarray*}
\end{proof}

Theorem \ref{lp1m} told us that $l$ increases proportionally with respect to $m-k$, with a ratio of $P_1$. Now we examine the bound of $P_1(k,p)$.

\begin{figure}[h]
\centering
\epsfig{file=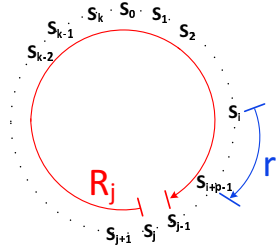, width=0.2\textwidth}
\caption{Illustration of Theorem 3.3}
\label{fig:illustration_theorem}
\end{figure}

\begin{theorem}
\label{thrm:totalsize}
In a random string model, $P_1(k,p)\leq \frac{p+1}{k+1}$.
\end{theorem}

\begin{proof}
Given any string $s=s_0s_1 \ldots s_k$, we concatenate $s_k$ and $s_0$ to form a ring as depicted in Figure \ref{fig:illustration_theorem}.  The ring can generate $k+1$ length-($k+1$) strings by starting at different positions: $R_j$=$s_js_{j+1}\ldots s_{(j+k) mod (k+1)}$ $j=0,1,\ldots,k$. Let $S_{k,p}=\{$  length-$(k+1)$ string whose first or last p-substring is the only minimum p-substring in it $\}$. We have $P_1(k,p)=|S_{k,p}|/4^{k+1}$.  Now we calculate at most how many $R_i$ strings belong to $S_{k,p}$.  Let $r$ be one of the minimum p-substrings among all of the p-substrings in $\{R_i\}$.  For any $R_i$, if $r$ is located inside $R_i$ (neither in the head nor the tail), then $R_i$ does not belong to $S_{k,p}$.  In total, there are $k-p$ $R_i$'s satisfying this condition.  So in these $k+1$ $R_i$ strings, at most $p+1$ of them can possibly belong to $S_{k,p}$. This gives us $P_1(k,p)\leq \frac{p+1}{k+1}$.
\end{proof}

\begin{corollary}
In a random string model, the total partition size is $O (pn)$.
\end{corollary}
\begin{proof}
According to Theorems \ref{thrm:partitiontotalsize} and \ref{thrm:totalsize},
$\frac{lk}{m}n+n < \frac{(m-k)n}{m}(p+1) + n < (p+1)n+n= O(pn)$.
\end{proof}

Since $p<<k$, the total partition size $O (pn)$ is far smaller than $\Theta (kn)$ in traditional partition methods.  In practice, $p$ is fixed as a small constant; thus the size becomes $\Theta (n)$.  In the following discussion, we present a stronger bound for the total partition size without this assumption.

\begin{theorem}
\label{thrm:pbound}
In a random string model, for any integer $a>0$, $P_1(k+a,p+a) \leq 2\cdot P_1(k,p)+\frac{p+2}{4^p}$.
\end{theorem}
\begin{proof}
Let $S_{k,p}=\{s||s|=k+1$, $s$' first or last p-substring is the only minimum p-substring in $s\}$, $S_{k,p}^*=\{s||s|=k+1$, $s$' first or last p-substring is one of the minimum p-substrings in $s\}$.  We have $P_1(k,p)=|S_{k,p}|/4^{k+1}$. Let $P_2(k,p)=|S_{k,p}^*|/4^{k+1}$.  Given a $k+a+1$ length string $t$ that belongs to $S_{k+a,p+a}$, consider the $k+1$ length string consisting of the first $k+1$ characters of $t$. Obviously it belongs to $S^*_{k,p}$, so we have $P_1(k+a,p+a)\leq P_2(k,p)$. Hence, we only need to prove $P_2(k,p)-P_1(k,p)\leq P_1(k,p)+\frac{p+2}{4^p}$.

Given a string $s=s_1s_2 \ldots s_{k+1}$ which belongs to set $S_{k,p}^*-S_{k,p}$, we build an injective mapping from $S_{k,p}^*-S_{k,p}$ to $S_{k,p}$. For the situation where the first p-substring of $s$ is one of the minimum p-substrings, let $s_r$ be the first character that is not $0$. Then we map $s=s_1s_2 \ldots s_{k+1}$ to $s'=s_1,\ldots, s_{r-1}, s_{r}-1, s_{r+1}, \ldots, s_{k+1}$. It is easy to see that $s'$ belongs to $S_{k,p}$, except two situations: p-substring $s_1s_2 \ldots s_p$ is (1) $00\ldots 0$ or (2) has only one $1$ while all other characters are $0$. These two situations have a probability of $\frac{p+1}{4^p}$ (detailed proof omitted due to space limit).  Similarly, for the other situation where the last p-substring of $s$ is one of the minimum p-substrings, we map $s=s_1s_2 \ldots s_{k+1}$ to $s''=s_1,\ldots, s_{r-1}, s_{r}-1, s_{r+1}, \ldots, s_{k+1}$, where $s_r$ is the last character that is not $0$. Then $s''$ belongs to $S_{k,p}$, except for the case that p-substring $s_{k-p+2}s_{k-p+3} \ldots s_{k+1}$ is $00\ldots 0$, whose probability is $\frac{1}{4^p}$. Hence, $|S^*-S|\leq |S|+\frac{p+2}{4^p}\cdot4^{k+1}$.  That is, $P_2(k,p)\leq 2\cdot P_1(k,p)+\frac{p+2}{4^p}$.
\end{proof}

Assuming $k=m/2$, $k<100$, $p<k/5$, we have
\begin{eqnarray*}
kl=k\cdot l(m, k, p)=k\cdot P_1(k,p)\cdot (m-k)\,\,\,\,&(Theorem\,\ref{lp1m})\\
<(2k\cdot P_1(k-p+5,5)+k\cdot\frac{7}{4^5})\cdot (m-k)&(Theorem\,\ref{thrm:pbound})\\
<(2\cdot\frac{k}{k-p+6}\cdot6+0.7)\cdot m/2\qquad\qquad\,\,\,&(Theorem\,\ref{thrm:totalsize})\\
<(12\cdot\frac{100}{86}+0.7)\cdot m/2<7.4m.\qquad\qquad
\end{eqnarray*}

Therefore, $\frac{kl}{m}n + n < 8.4n$, which is much better than $\Theta (kn)$.

\subsection{Largest Partition Capacity}  Since MSP has to load/hash each partition into main memory, the largest\emph{ partition capacity}, defined as the maximum number of distinct k-mers contained by a partition, determines the peak memory. We study its upper bound and lower bound in a random string model.

\begin{theorem}
\label{thrm:peakmemory}
In a random string model, the maximum percentage of distinct k-mers covered by one p-substring is bounded by $\frac{3k}{4^{p+1}}$, when $p\geq 2$.
\end{theorem}
\begin{proof}
In a random string model, each symbol has equal opportunity to appear in each position of short reads.  The probability of observing any length-$m$ string is equal.  As the smallest p-substring defined by lexicographical order, the partition built on the p-substring $00\ldots 0$ has the largest number of distinct k-mers.

Let $\alpha(k,p)$ denote the percentage of distinct k-mers covered by p-substring $00\ldots 0$.  This percentage is not related to $m$. For fixed $p$, there are two situations for a k-mer to have a p-substring $00\ldots 0$: (1) the first $k-1$ characters have a p-substring $00\ldots 0$, or (2) the last p-substring is the only $00\ldots 0$ in this k-mer. Note that in the later situation the first $k-p-1$ characters must not have a p-substring $00\ldots 0$ and the $(k-p)^{th}$ character must not be $0$.  This gives us
\begin{equation}
\alpha(k,p)=\alpha(k-1,p)+(1-\alpha(k-p-1,p))\frac{3}{4}\cdot\frac{1}{4^p} \nonumber \end{equation}
Obviously, $\alpha(p,p)=\frac{1}{4^p}$ and $\alpha(k,p)\geq \alpha(k-1,p)$. Thus
\begin{eqnarray*}\label{equ:31}
\alpha(k,p)=\alpha(k-1,p)+(1-\alpha(k-p-1,p))\frac{3}{4}\cdot\frac{1}{4^p}\,\,\,\\
<\alpha(k-1,p)+\frac{3}{4^{p+1}}<\frac{1}{4^p}+(k-p)\cdot \frac{3}{4^{p+1}}\\
<\frac{3k}{4^{p+1}}, \mathrm{when}\, p\geq 2. \qquad\qquad\qquad\qquad\quad\quad\,
\end{eqnarray*}
\end{proof}

We can further establish a lower bound,
\begin{eqnarray*}
\alpha(k,p)=\alpha(k-1,p)+(1-\alpha(k-p-1,p))\frac{3}{4}\cdot\frac{1}{4^p}\\
>\alpha(k-1,p)+(1-\alpha(k,p))\cdot \frac{3}{4^{p+1}}\qquad\quad\,\\
>\frac{1}{4^p}+(k-p)\cdot (1-\alpha(k,p))\cdot \frac{3}{4^{p+1}}.\qquad\,\,\,
\end{eqnarray*}
From above, if $4<p<k/5$, we have
\begin{eqnarray*}
\frac{2k}{4^{p+1}}<\alpha(k,p)<\frac{3k}{4^{p+1}}.
\end{eqnarray*}

\begin{figure}[h]
\centering
\epsfig{file=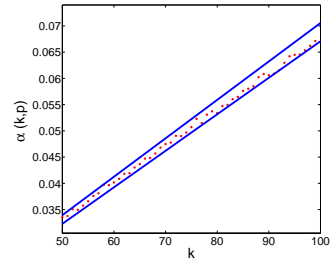, width=0.25\textwidth}
\vspace{-2mm}
\caption{The Bounds of $\alpha$(k,p)}
\vspace{-2mm}
\label{fig:maxperc}
\end{figure}

Figure \ref{fig:maxperc} depicts the bounds for the expected percentage of k-mers covered by the largest partition (corresponding to a p-substring $00\ldots 0$) with respect to different $k$ values. $p$ is set at $5$.  The result shows that the bounds we have proved are good: When $k$ changes from $50$ to $100$, the maximum percentage of distinct k-mers covered by one minimum p-substring (the largest partition) is quite close to the lower and upper bounds we provided.

To calculate the entire distribution of partition capacities (the number of distinct k-mers covered by each p-substring) in a random string model, we develop an efficient quadratic-time algorithm ($O(m^2)$, see the Appendix). Using this algorithm, we do not need to use costly simulation to estimate the partition capacity.

\begin{figure}[h]
\centering
\epsfig{file=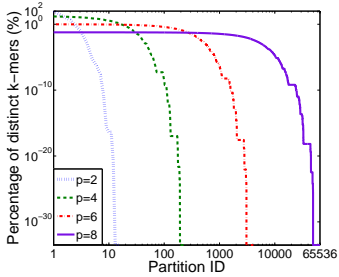, width=0.25\textwidth}
\vspace{-2mm}
\caption{Expected Partition Capacity Distribution}
\vspace{-2mm}
\label{fig:estimation-distribution}
\end{figure}

Figure \ref{fig:estimation-distribution} shows the expected distribution of partition capacities with respect to different minimum substring lengths, assuming that 4 bases A, C, G, T appear with equal probability and k-mer length is 59. Here the p-substrings are sorted according to the percentage of $k$-mers they cover. The figure uses logarithm on both axes. The result shows a property: when $p$ increases, there is a plateau where many $p$-substrings cover a similar percentage of distinct k-mers. We can conclude that there is no extremely memory consuming partition when $p$ is not very small. Furthermore, the peak memory of MSP can be fully controlled by $p$.

\section {Reverse Complements}
\label{sec:reverse}
DNA sequences can be read in two directions: forwards and backwards with each symbol changed to its Watson-Crick complements ($A \leftrightarrow T$  and $C \leftrightarrow G$). They are called \emph{reverse complement} and considered equivalent in bioinformatics.  Most sequencing techniques extract short reads in either direction. In an assembly processing, each sequence should be read twice, once in the forward direction and then in the reverse complement direction.

Reverse complement is not an issue for bucket partitioning: when a k-mer is read into memory, a reverse complement can be built online.  It becomes tricky for minimum substring partitioning since MSP intends to compress consecutive k-mers together if they share the same minimum $p$-substring.  Unfortunately, their reverse complements might not share the same minimum $p$-substring.  This forces us to generate the reverse complement explicitly for each short read, which will double the I/O cost.
\begin{definition}

[Minimum Substring with Reverse Complements]
\label{def:reverse}
Given a string $s$, a length-p substring $t$ of $s$ is called the minimum p-substring of $s$, if $\forall s'$, $s'$ is a length-p substring of $s$ or $s$' reverse complement, s.t., $t\leq s'$ ($\leq $ defined by lexicographical order).
\end{definition}

Definition \ref{def:reverse} redefines minimum substring by considering the reverse complement of each k-mer.  With this new definition, we need not output reverse complements explicitly, nor change the minimum substring partitioning process.  In the following discussion, if not mentioned explicitly, we will ignore this problem.

\section{Algorithms}
\label{sec:algorithms}
In this section, we describe the detailed algorithm to build a de Bruijn graph. It consists of three steps: Partitioning, Mapping and Merging.  Each step is performed by a program that takes an on-disk representation of input and produces a new on-disk representation of output.  The input of the first step is the raw short read sequences and the output of the last step is a sequence of id's mapped to the k-mers in short read sequences, in the same order; the duplicate k-mers shall have the same id.

\subsection{Partitioning}
The first step is to partition short reads using MSP.  A straightforward approach is as follows: (1) given a short read $s$, slide a window of width $k$ through $s$ to generate k-mers, (2) for each k-mer, calculate its minimum p-substring, (3) find super k-mers in $s$ (adjacent k-mers sharing the same minimum p-substring).  This method has to calculate the minimum p-substring of every $k$-mer.  Each $k$-mer needs $(k-p+1)$ p-substring comparisons.  Let $m$ be the length of $s$.  In total, this approach needs to perform $(k-p+1)*(m-k+1)$=$\Theta (mk)$ p-substring comparisons.

The above solution does not leverage the overlaps among adjacent k-mers.  When the $k$-size window slides through $s$, we can maintain a priority queue on $p$-substrings in the window.  Each time, when we slide the window one symbol to the right, we drop the first $p$-substring in the previous window from the queue and add the last $p$-substring of the current window into the queue.  Since the number of $p$-substrings in a window is $k-p+1$ and there are $m-p+1$ p-substrings in $s$, the number of p-substring comparisons is $O((m-p+1)\log(k-p+1))$=$O(m\log k)$.

While the priority queue is theoretically good, the overhead introduced by the queue structure could be high.  We thus introduce a simple scan algorithm, as described in Algorithm \ref{algo:partition}.  Algorithm \ref{algo:partition} first scans the window from the first symbol to find the minimum p-substring, say min\_s, and the start position of min\_s, say min\_pos.  Then it slides the window towards right, one symbol each time, till the end of the short read.  After each sliding, it tests whether min\_pos is still within the range of the window. If not, it re-scans the window to get the new min\_s and min\_pos. Otherwise, it tests whether the last p-substring of the current window is smaller than the current min\_s. If yes, this last p-substring is set as the new min\_s and its start position as the new min\_pos.  As analyzed in the previous section, adjacent k-mers likely have the same minimum p-substring.  Therefore, it needs not to re-scan the window very often.  Although the worst case time complexity is $O(mk)$ $p$-substring comparisons, Theorem \ref{theorem:efficiency} shows it could be more efficient in practice since the average number of breaks is small.

\begin{algorithm}
\caption{SimpleScan}
\label{algo:partition}
\begin{algorithmic}
\STATE Input: String $s=s_1 s_2 \ldots s_m$, integer $k, p$.
\STATE min\_s = the minimum p-substring of $s[1, k]$
\STATE min\_pos = the start position of min\_s in $s$

\FORALL {$i$ from $2$ to $m-k+1$}
    \IF {$i$ $>$ min\_pos}
        \STATE min\_s = the minimum p-substring of $s[i, i+k-1]$
        \STATE min\_pos = the start position of min\_s in $s$
    \ELSE
        \IF {the last p-substring of $s[i, i+k-1]$ $<$ min\_s}
            \STATE min\_s = the last p-substring of $s[i, i+k-1]$
            \STATE min\_pos = the start position of min\_s in $s$
        \ENDIF
    \ENDIF
\ENDFOR
\end{algorithmic}
\end{algorithm}

\begin{theorem}
\label{theorem:efficiency}
Given an $m$-length string, assume minimum substring partitioning divides $s$ into $l+1$ substrings. Algorithm \ref{algo:partition} needs at most $\Theta (m+lk)$ p-substring comparisons.
\end{theorem}
\begin{proof}
Algorithm \ref{algo:partition} shows that min\_s and min\_pos change under two conditions: (1) $i$ $>$ min\_pos, or (2) the last p-substring of $s[i, i+k-1]$ $<$ min\_s. Under the first condition, it re-scans the k-mer $s[i, i+k-1]$, which introduces $k-p+1$ p-substring comparisons. Under the second condition, it compares the last p-substring of $s[i, i+k-1]$ with the current min\_s, which involves $1$ p-substring comparison.  Since the string $s$ is broken into $l+1$ substrings, min\_s and min\_pos changes for $l$ times.  If all these $l$ changes are due to the first condition, the total number of k-mer scans is $l+1$, including the initial scan of the first k-mer. This results in $(k-p+1)*(l+1)$ p-substring comparisons.  Within each of these $l+1$ substrings, it needs $n_t-1$ p-substring comparisons to test the second condition, where $n_t$ is the number of k-mers within the substring $s[i_t, j_t]$. For all $l+1$ substrings, the total number of p-substring comparisons due to this test is $\Sigma_{t=1}^{l+1} (n_t-1) = m-k-l$. Therefore the total number of p-substring comparisons of Algorithm \ref{algo:partition} is bounded by $(k-p+1)*(l+1)+(m-k-l)=m+lk-pl-p+1$, which is $\Theta (m+lk)$.
\end{proof}

\begin{definition}[Wrapped Partitions]
Given a string set $\{s_i\}$,  a hash function $H$, the number of partitions $t$, for any k-mer $s_{i,j}$, minimum substring partition wrapping assigns $s_{i,j}$ to the $H(min_{p}(s_{i,j})) \mod t$ partition.
\end{definition}
\vspace{-0mm}
\begin{figure}[h]
\centering
\epsfig{file=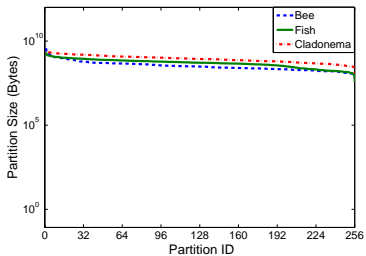, width=0.3\textwidth}
\vspace{-2mm}
\caption{Partition Wrapping}
\vspace{-0mm}
\label{fig:wrapping}
\end{figure}

Since each $p$-substring corresponds to one partition,  the total number of partitions in MSP is equal to $4^p$.  When $p$ increases, the number will increase exponentially.  To counter this effect, one can introduce a hash function to wrap the number of partitions to any user-specified partition number. In this case, each partition generated from a $p$-substring is randomly included in a wrapped partition.  The variance of partition sizes will likely decrease.  Figure \ref{fig:wrapping} shows the distribution of partition size when $p=10$ and the number of wrapped partitions is set to $256$.  The number of partitions is the same as that of $p=4$ without wrapping.  In comparison with Figure \ref{fig:distribution}, the partition size distribution is more uniform.

\subsection{Mapping}
In this step, each distinct k-mer is mapped to a unique integer id as its vertex id in the de Bruijn graph.  A straightforward solution is to process each partition one by one.  For each partition, insert k-mers into a hash table.  Whenever there is a k-mer that does not exist in the table, a new id is assigned to it.  The starting id of k-mers in one partition is the maximum k-mer id of the previous partition plus one. After one partition is processed, a disk file (called \emph{id file}) is created, the entries in hash table are written to that file. This approach works well for the mapping step. However it will cause a serious problem in the merging step.

\begin{figure}[h]
\centering
\epsfig{file=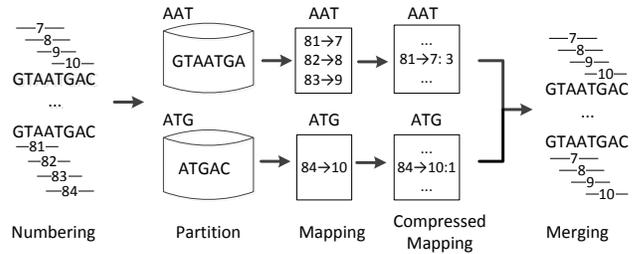, width=0.48\textwidth}
\caption{ID Replacement and Merging}
\vspace{-1mm}
\label{fig:replacement}
\end{figure}

In the merging step, we scan the short reads again to build edges for adjacent k-mers.  For each pair of adjacent k-mers, we need to locate their ids from their corresponding id files. Considering that the id files are as big as the partition files, it will cause a lot of I/O and seriously slow down the process.  In order to solve this problem, we develop an id replacement strategy, as depicted in Figure \ref{fig:replacement}.  During the partitioning step, each k-mer is assigned an integer id. The id's are assigned increasingly from $1$.  The same k-mer in different short reads receives different id's. Our goal is to replace them with the first id it receives. For each partition, we create a hash table in memory. Whenever we see a new k-mer, we first look up the hash table to see if it exists: if yes, we write a replacement record into the id replacement file, indicating that the current k-mer is a duplicate and we have to replace its pre-assigned id with the id associated with its first occurrence.

Figure \ref{fig:replacement} shows an example of the id replacement process.  Assume $k=5$, $p=3$, and $GTAATGAC$ occurs in two different short reads.  In the beginning, each k-mer in two $GTAATGAC$ is assigned a unique id, e.g, $7-10$ and $81-84$, respectively.  Sequences $GTAATGA$ is sent to the $AAT$ partition, while $ATGAC$ is sent to the $ATG$ partition.  During this process, the k-mer $81$ is mapped to the k-mer $7$, $82$ to $8$, and $83$ to $9$, while the k-mer $84$ is mapped to the k-mer $10$.  To compress the replacement file, we write the replacement records as a range instead of multiple individual records.  For the example shown in Figure \ref{fig:replacement}, one can just output a range record, $81\rightarrow 7: 3$, meaning the 3 consecutive id's starting at 81 will be replaced by 3 consecutive id's starting at 7.  Range compression is quite effective since there are many long overlaps in short reads.  According to our experiments, this kind of compression reduces the size of id replacement files to that of the original short read file.

\subsection{Merging}

After obtaining the id replacement files, the last step is merging.   In this step, we merge all the replacement files to generate a sequence of id's that map to the original short reads, in the same order.  We first open all the replacement files with each file header pointing to the first id replacement record of the corresponding file.   Since all the files are already naturally sorted in increasing order by the first entry (the pre-assigned id's to be replaced) of replacement record,  we can find the minimum id to be replaced in the current filer headers.  We write its replacing id,  move to the next replacement record, and iterate.  After this process, we get a sorted sequence of id's corresponding to k-mers in the short read dataset, in the same order.  This actually forms a disk-based de Bruijn graph.  The last step in Figure \ref{fig:replacement} shows this process, where two duplicated $GTAATGAC$ sequences receive the same id sequence.  This disk-based graph can either be distributed across multiple machines, or compressed \cite{zerbino2008velvet} and loaded into memory.

\section{Experiments}
\label{sec:experiments}
In this section, we present experimental results to illustrate the memory efficiency, effectiveness and important properties of the minimum substring partitioning method on four large real-life datasets: cladonema, bumblebee, fish, and bird. (1) We first analyze the efficiency of our graph construction algorithm in terms of memory and time cost, and compare it with two well-known open-source assembly programs, Velvet \cite{zerbino2008velvet} and SOAPdenovo \cite{li2010novo}. (2) The performance of MSP and two traditional partition/merge algorithms, \HPartition, \BPartition, are compared in terms of partition size and runtime. (3) We change different parameter settings to demonstrate important properties of MSP.  All the experiments, if not specifically mentioned, are conducted on a server with 2.40GHz Intel Xeon CPU and 512 GB RAM.

\subsection{Data Sets}

Four real-life short reads dataset are used to test our algorithms. The first one is the sequence data of Cladonema provided by our collaborators. The bee, fish and bird datasets are available via \url{http://gage.cbcb.umd.edu/data/Bombus_impatiens}, \url{http://bioshare.bioinformatics.ucdavis.edu/Data/hcbxz0i7kg/Fish}, and \url{http://bioshare.bioinformatics.ucdavis.edu/Data/hcbxz0i7kg/Parrot/BGI_illumina_data}, respectively. Table $1$ shows some basic facts.

\begin{table}[hbp]
\label{tbl:datasets}
\centering
\begin{tabular}{lcccc}
\hline
 &Cladonema &Bee    &Fish &Bird \\ \hline
Size(GB) &258.7   &93.8   &137.5 & 106.8\\
Avg Read Length(bp) &101    &124   &101 & 150 \\
\# of Reads(million) &894   &303  &598 & 323 \\
\hline
\end{tabular}
\caption{Datasets: Cladonema, Bombus impatiens(bee), Lake Malawi cichlid(fish), and Budgerigar(bird)}
\vspace{-5mm}
\end{table}

\subsection{Efficiency}
We first conduct experiments to compare MSP with two real sequence assembly programs on de Bruijn graph construction: Velvet \cite{zerbino2008velvet}, a classic de Bruijn graph based assembler, and SOAPdenovo \cite{li2010novo}, a highly optimized and leading assembler.  For all the experiments, we set the k-mer length to 59 \cite{butler2008allpaths}. For MSP, we partition the short reads into 1,000 wrapped partitions with the minimum substring length p being 12. SOAPdenovo is optimized to support multithreading, we use 2 threads here to illustrate its advantage. Both Velvet and MSP use 1 thread.  The 8-thread version of SOAPdenovo can roughly achieve the same runtime as MSP.  However, its peak memory consumption is still the same as its 2-thread version.

\begin{figure}[h]
\centering
\subfigure[\small\textit{Peak Memory}]{\label{fig:build-memory}\includegraphics[scale=0.35]{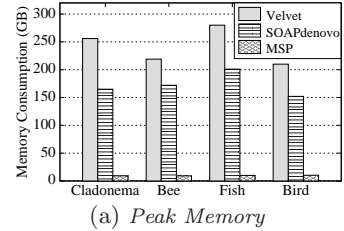}}
\subfigure[\small\textit{Running Time}]{\label{fig:build-time}\includegraphics[scale=0.35]{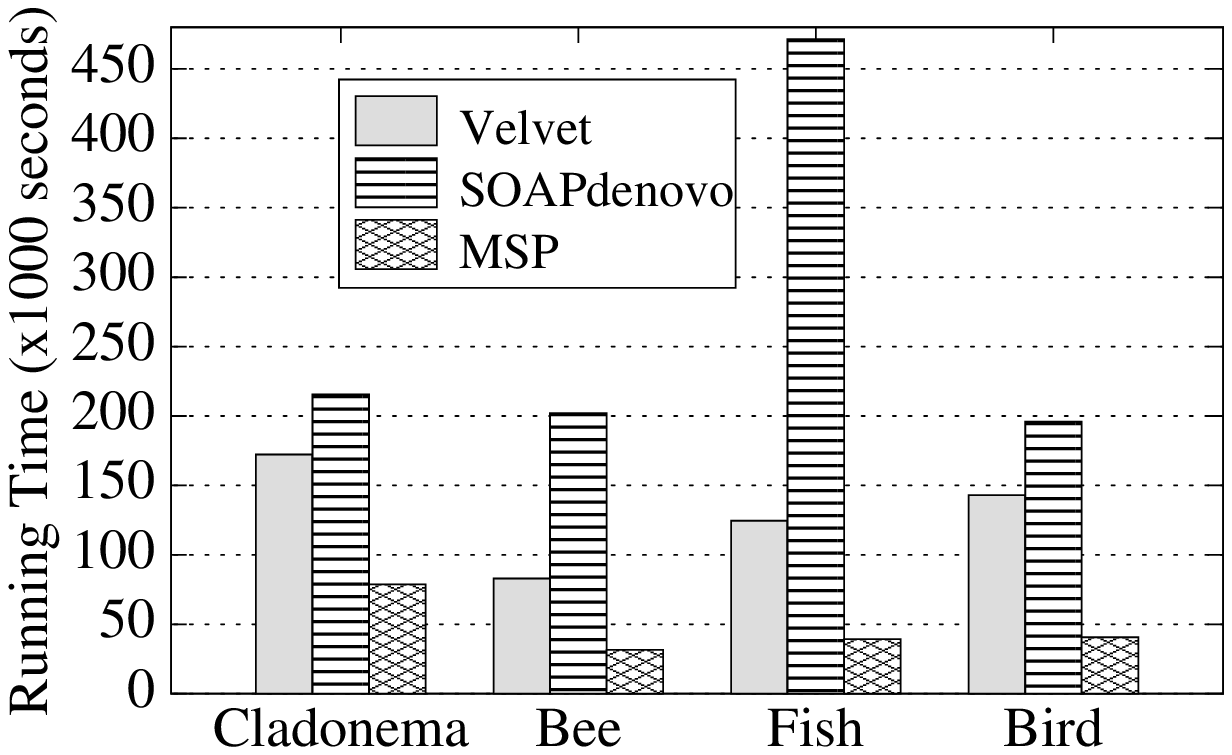}}
\caption{Velvet, SOAPdenovo, and MSP}
\label{fig:efficiency}
\end{figure}

Figure \ref{fig:efficiency} demonstrates that MSP outperforms Velvet and SOAPdenovo in terms of memory usage and running time.  For large datasets, Velvet and SOAPdenovo easily consume more than 150G memory, while our method can complete the task with less than 10G memory, an order of magnitude reduction of memory usage\footnote{If the entire de Bruijn graph needs to be loaded in main memory, we have routines available that consume 20-30\% of the memory that SOAPdenovo needs.}.

\subsection{Effectiveness}
We then conduct experiments to compare MSP with other partition/merge algorithms, \HPartition\ and \BPartition.  For all the three methods, we set the k-mer length to 59 and partition short reads into 1,000 partitions.  For MSP, we set the minimum substring length p at 12.  For \BPartition, we use the last 4 symbols to partition k-mers.  All the three algorithms use the similar amount of memory (around 10 GB). Figures \ref{fig:total-size} and \ref{fig:runtime} show the maximum disk space usage and the total running time.

\begin{figure}[h]
\centering
\subfigure[\small\textit{Maximum Disk Space Usage}]{\label{fig:total-size}\includegraphics[scale=0.35]{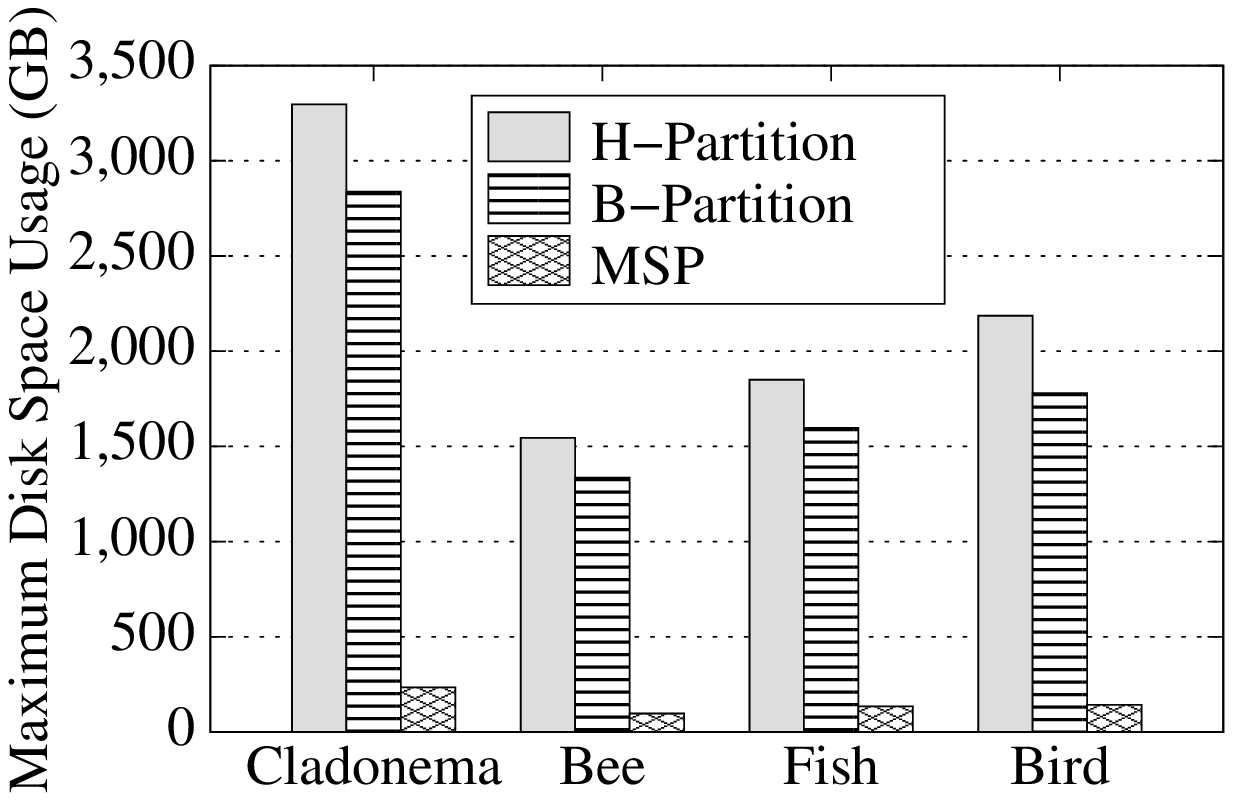}}
\subfigure[\small\textit{Running  Time}]{\label{fig:runtime}\includegraphics[scale=0.35]{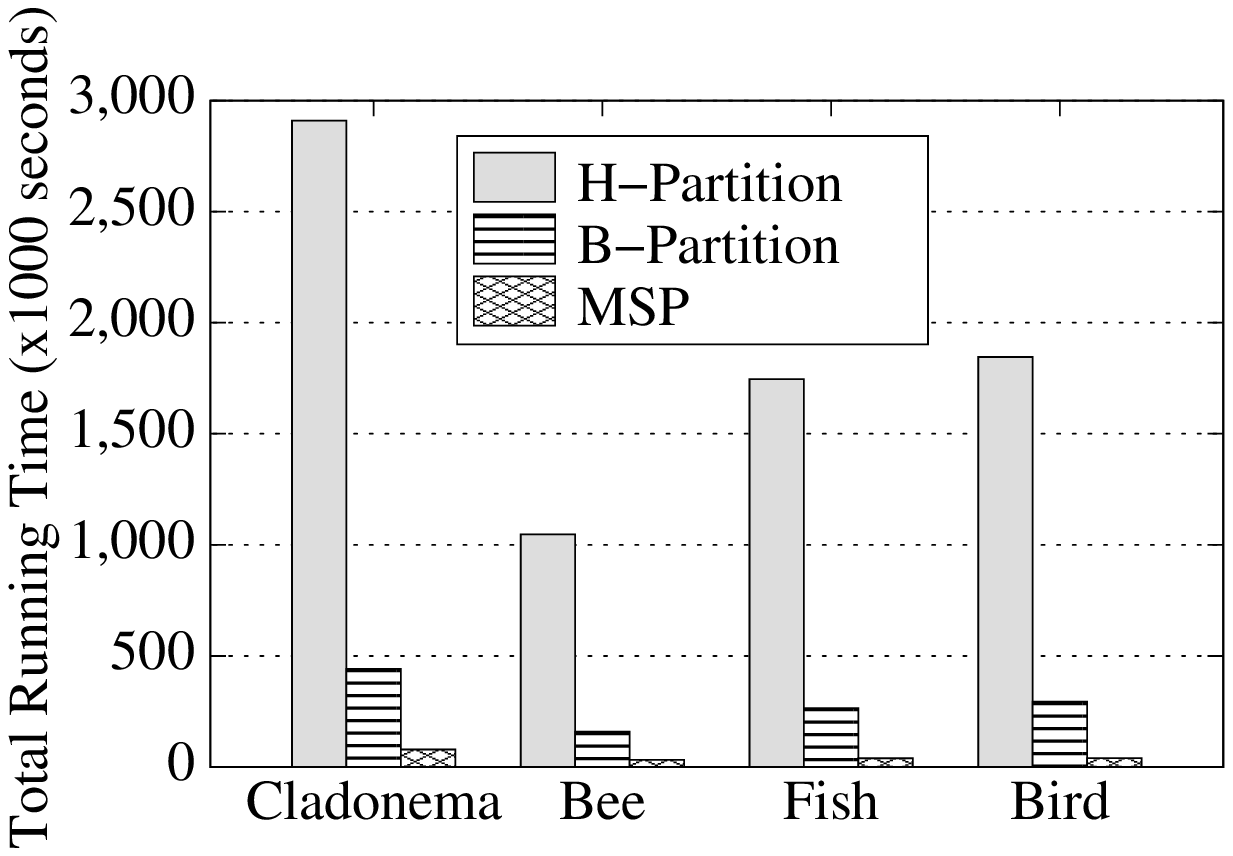}}
\caption{ \HPartition, \BPartition, MSP}
\label{fig:performance}
\end{figure}

Figure \ref{fig:performance} shows MSP outperforms the two baseline methods: compared with \BPartition, MSP can reduce the maximum disk space usage by 10-15 times and reduce the total execution time by 8-10 times. \BPartition\ was adopted by out-of-core algorithms such as \cite{Kund+10}. It implies that MSP is better than the classic approach that does not leverage the overlaps among data records.  \HPartition's overall performance is the worst since it needs multiple disk scans and sorts.

\begin{figure}[H]
\centering
\epsfig{file=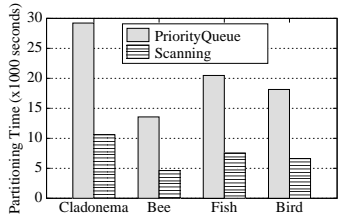, width=0.25\textwidth}
\vspace{-5mm}
\caption{SimpleScan vs. Priority Queue}
\label{fig:scan-method}
\end{figure}

We then illustrate the advantage of using a scanning method (Algorithm \ref{algo:partition}) over a priority queue approach in the partitioning step. Here we set $k$ at 59, $p$ at 12 and partition the short reads into 1,000 wrapped partitions. Figure \ref{fig:scan-method} shows that the simple scanning method in Algorithm \ref{algo:partition} is around 2 times faster than the priority queue approach.  Similar results were observed for other settings of $p$.

\subsection{Scalability}
We then conduct experiments to test the scalability of MSP. We vary the data size by randomly sampling the Clado-nema dataset. For MSP, we partition the short reads into 1,000 wrapped partitions with $p$ set at 10.

\begin{figure}[h]
\centering
\epsfig{file=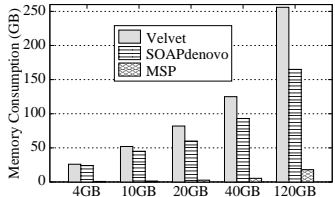, width=0.25\textwidth}
\vspace{-2mm}
\caption{Scalability: Peak Memory}
\label{fig:build-memory}
\end{figure}

\begin{figure}[h]
\centering
\epsfig{file=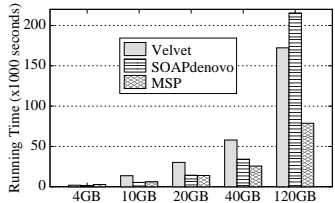, width=0.25\textwidth}
\caption{Scalability: Running Time}
\vspace{-2mm}
\label{fig:build-time}
\end{figure}

Figures \ref{fig:build-memory} and \ref{fig:build-time} show that all the three algorithms scale linearly in terms of peak memory consumption and running time. MSP performs the best.

\begin{figure*}[t!]
\centering
\subfigure[\small\textit{Peak Memory}]{\label{fig:peak-memory-p}\includegraphics[scale=0.18]{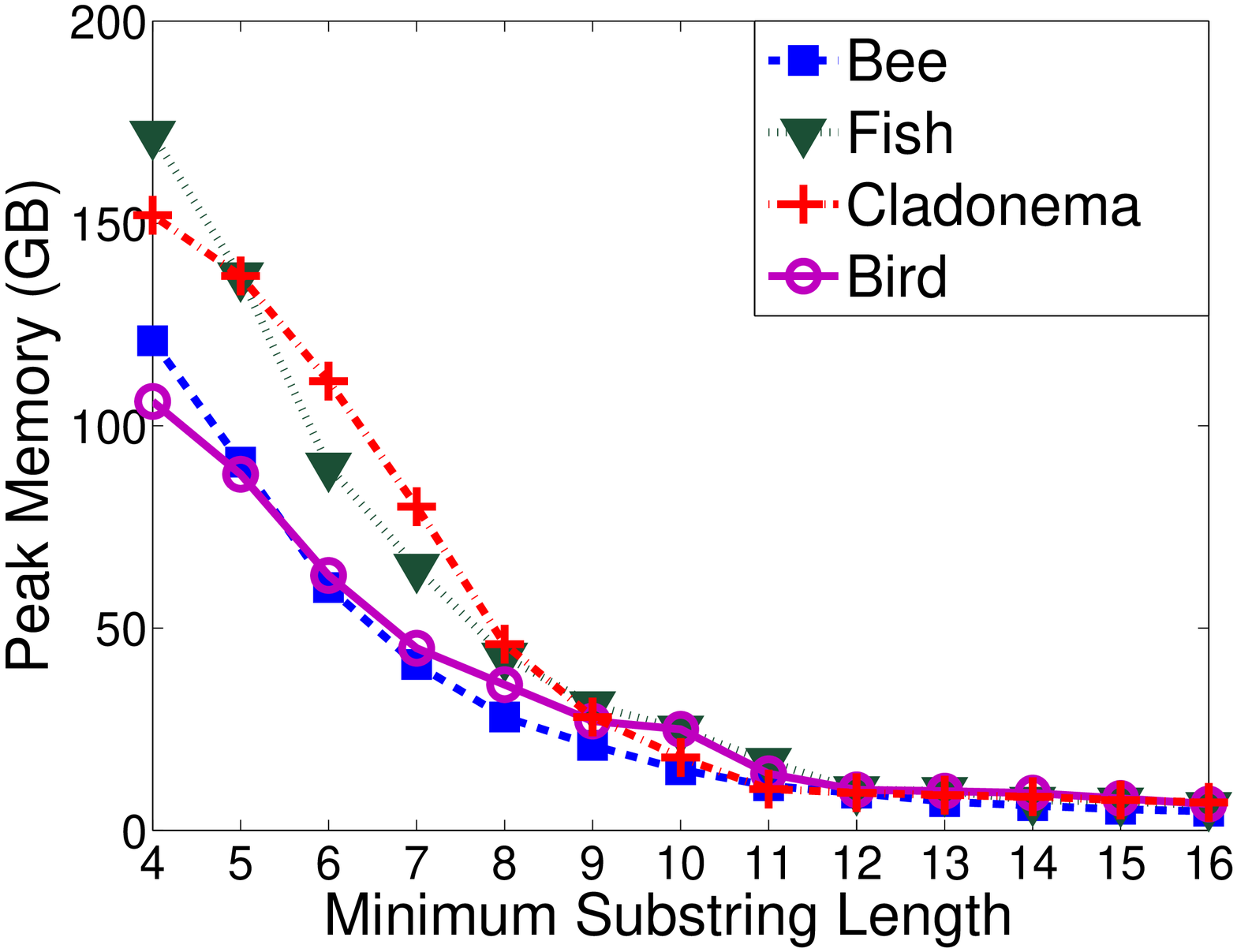}}
\subfigure[\small\textit{Total Partition Size }]{\label{fig:partition-length-p}\includegraphics[scale=0.18]{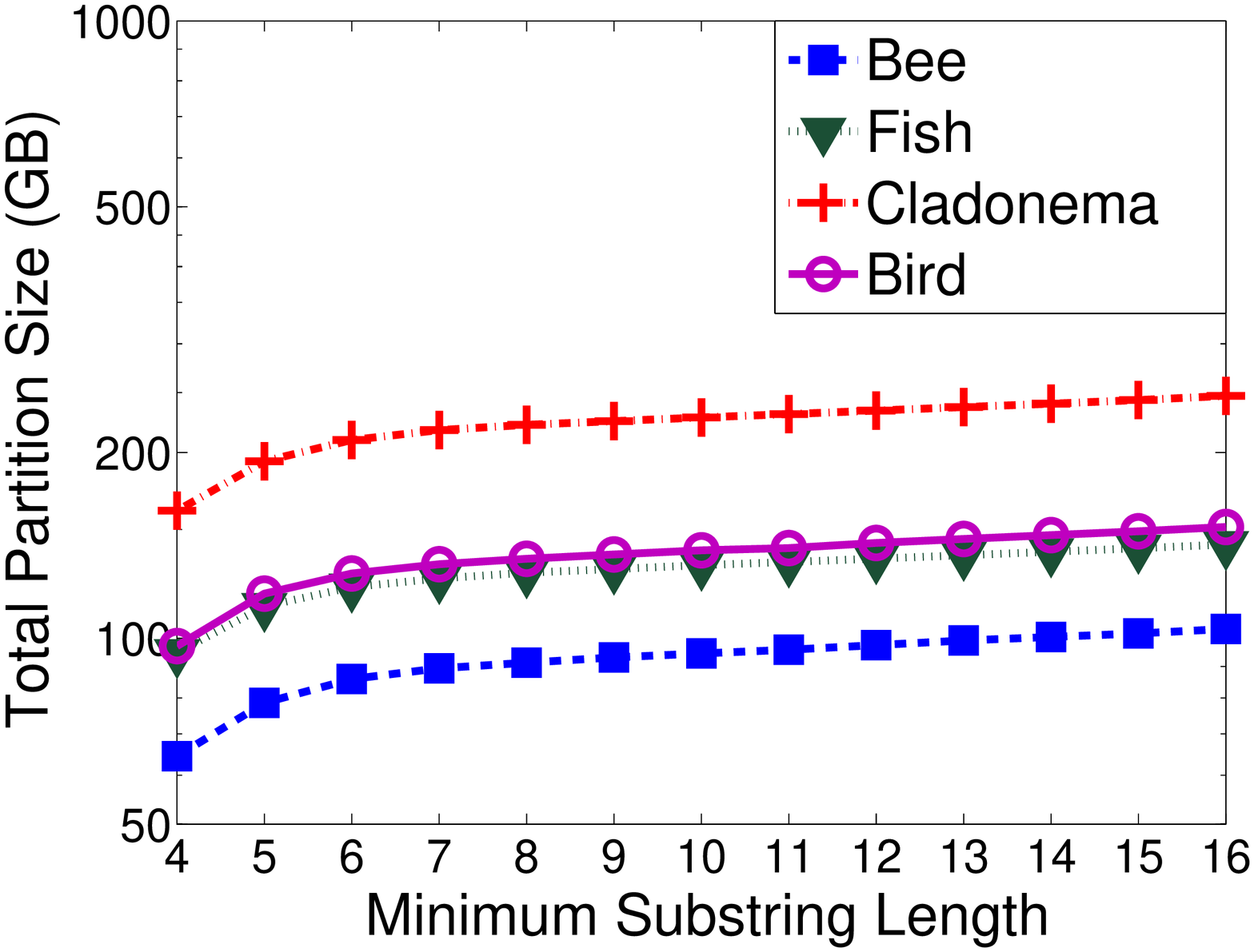}}
\subfigure[\small\textit{Running Time}]{\label{fig:runtime-length-p}\includegraphics[scale=0.18]{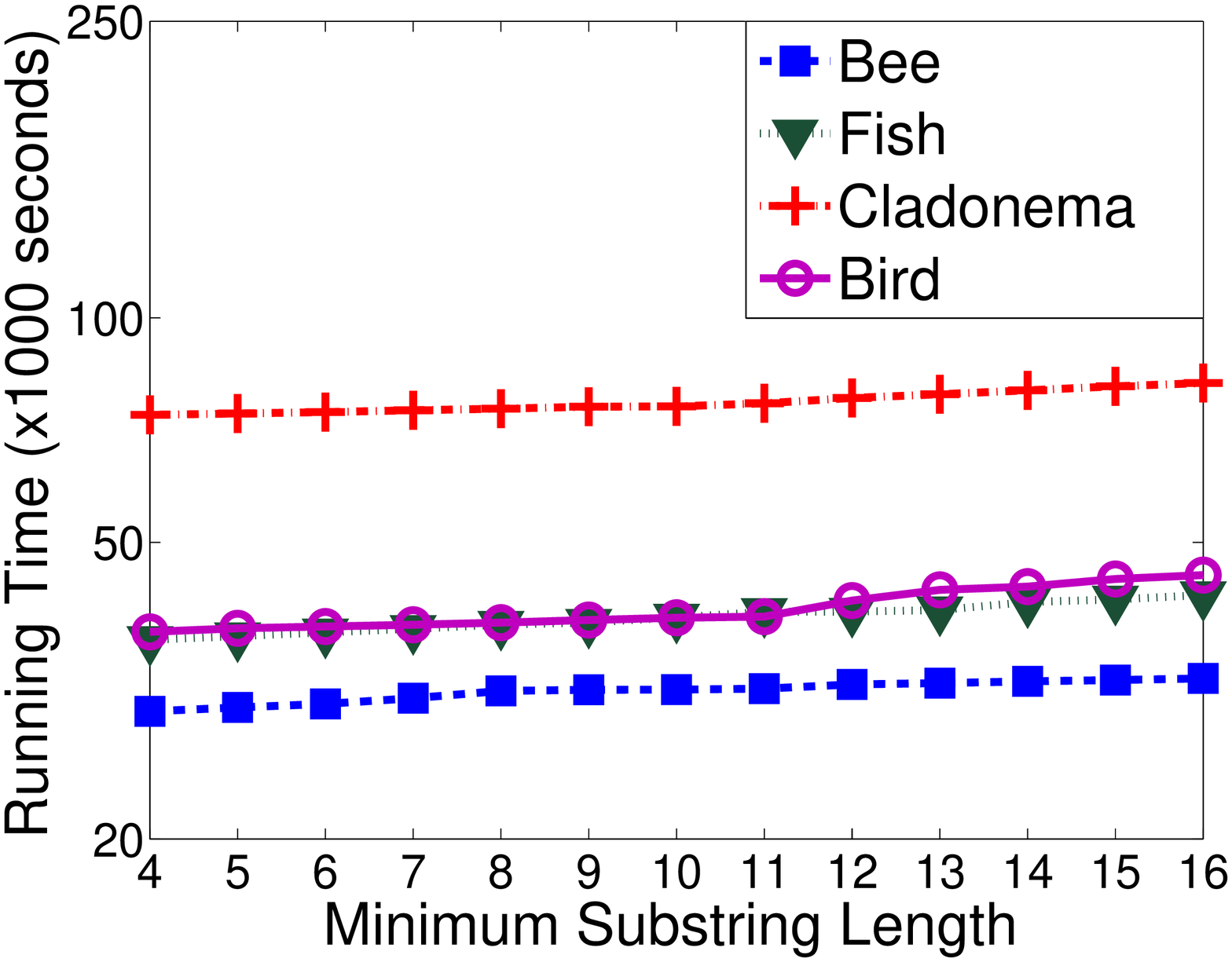}}
\caption{Varying Minimum Substring Length $p$}
\label{fig:plength}
\end{figure*}

\begin{figure*}[t!]
\centering
\subfigure[\small\textit{Peak Memory}]{\label{fig:peak-memory-k}\includegraphics[scale=0.18]{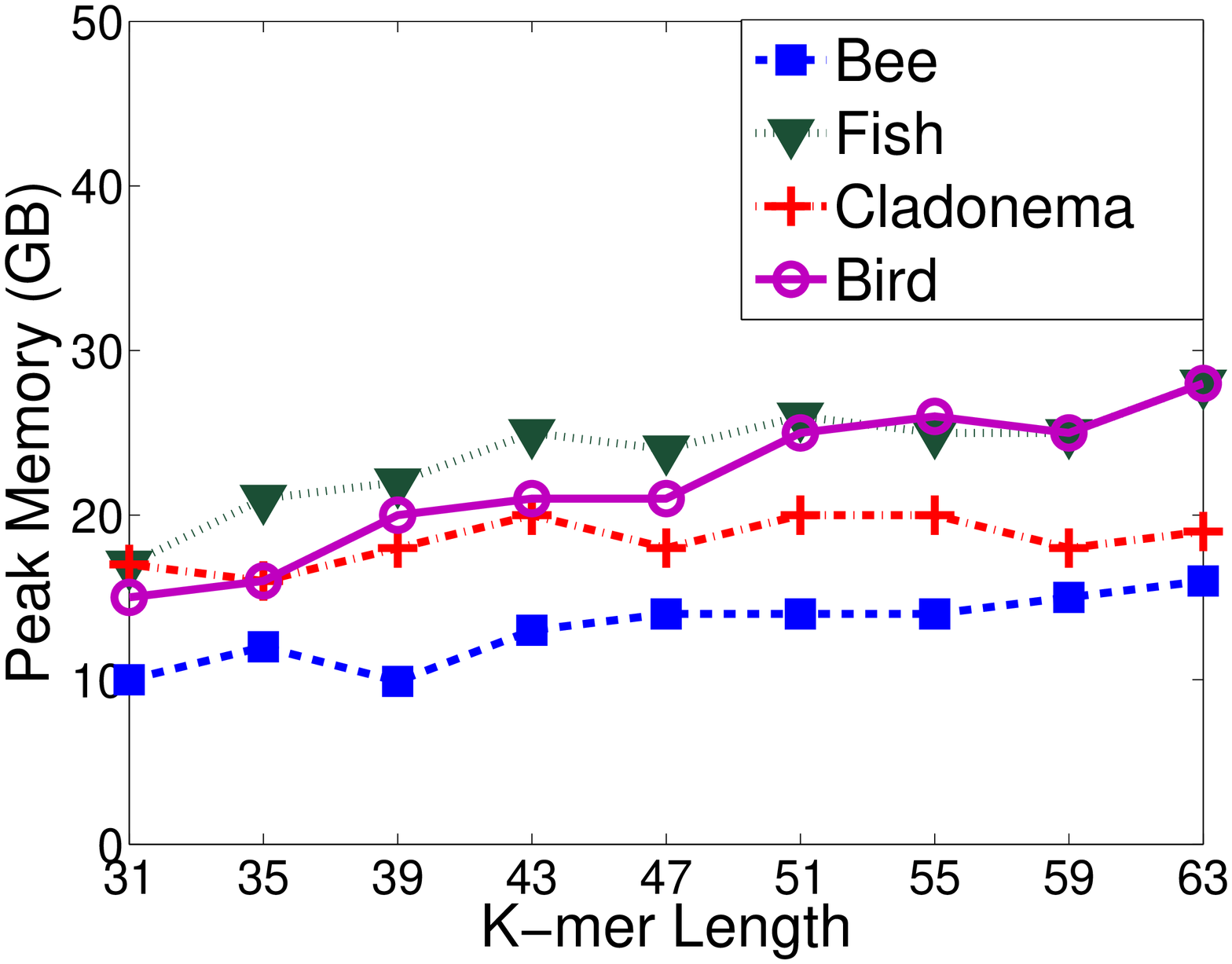}}
\subfigure[\small\textit{Total Partition Size }]{\label{fig:partition-length-k}\includegraphics[scale=0.18]{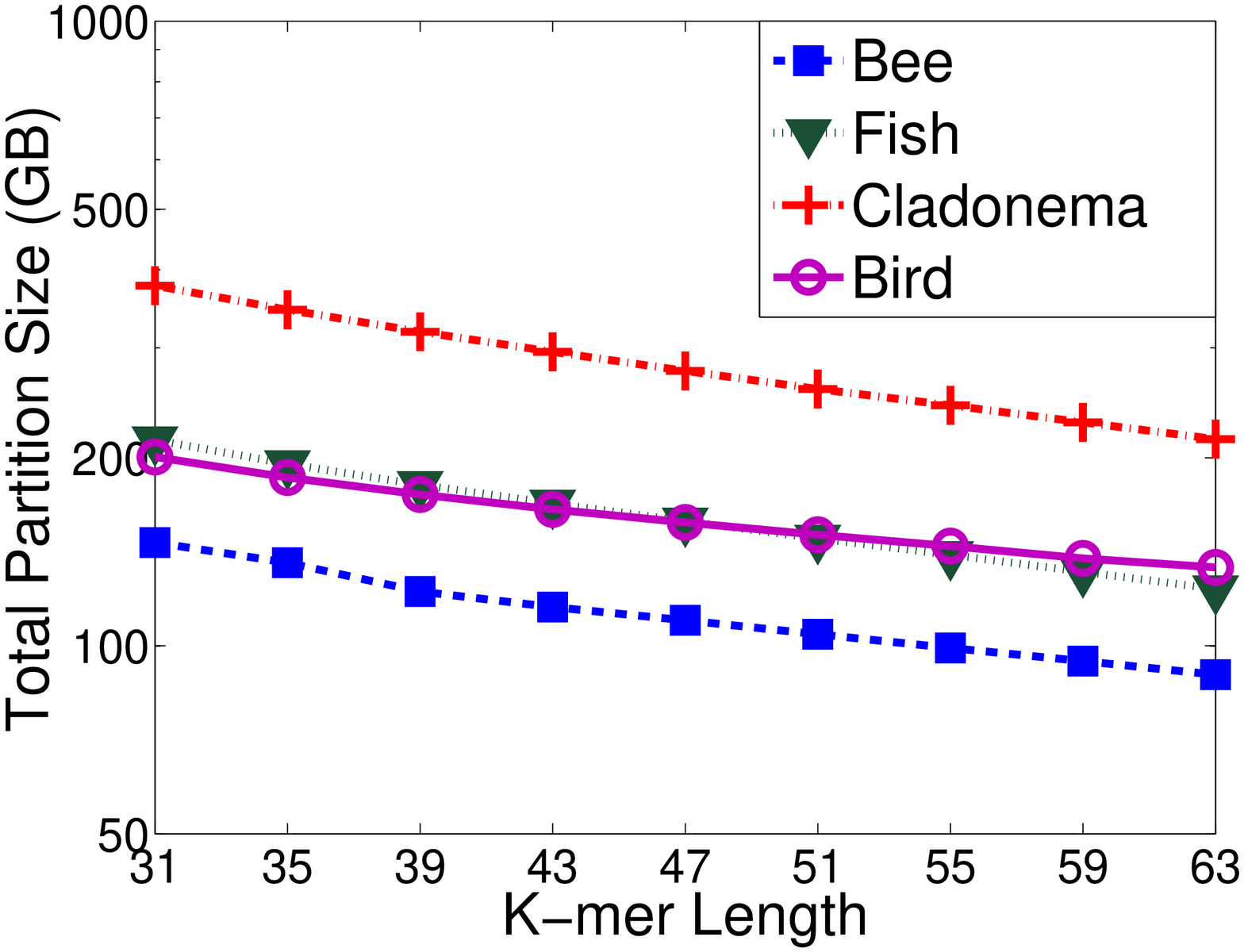}}
\subfigure[\small\textit{Running Time}]{\label{fig:runtime-length-k}\includegraphics[scale=0.18]{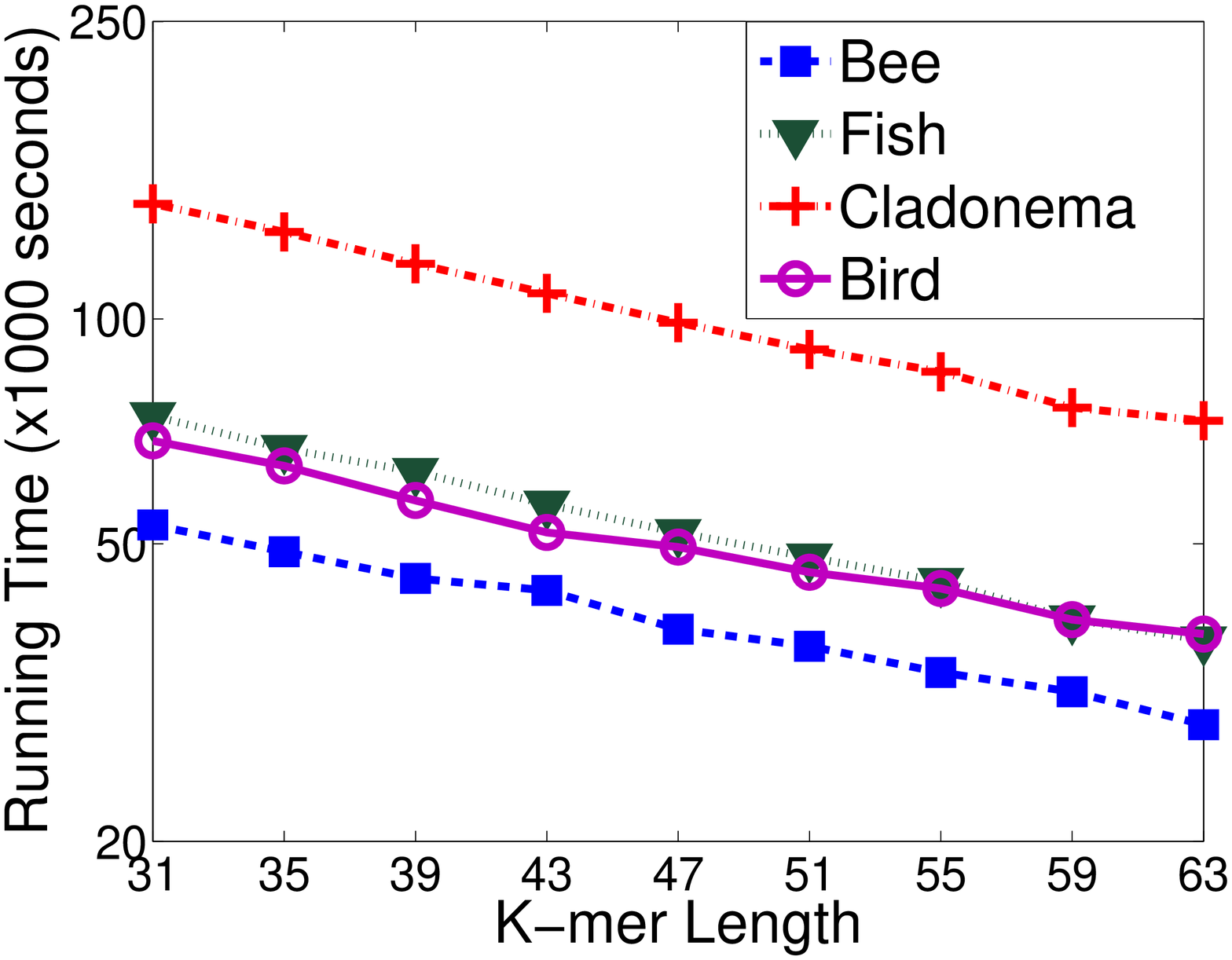}}
\caption{Varying K-mer Length $k$}
\label{fig:klength}
\end{figure*}

\subsection{Properties of MSP}
Next we conduct experiments to illustrate the properties of minimum substring partitioning. Figure \ref{fig:plength} shows the change of peak memory, partition size, and running time with respect to varying length of minimum substring.  Here, we set the k-mer length at 59 and partition short reads into 1,000 wrapped partitions. It shows that the peak memory will decrease significantly when the minimum substring length is increased. The total partition size and the running time will slightly increase.  Both increases are negligible, indicating that MSP is very effective in reducing memory consumption without affecting the runtime performance.

We then fix $p$ at 10, the number of partitions at 1,000, and vary the length of k-mers.  Figure \ref{fig:klength} shows the change of peak memory, partition size, and running time with respect to  k-mer length.  It shows that the peak memory increases slowly together with $k$. It is also observed that increasing $k$ will reduce the total partition size and the running time.  There are two effects inside.  Given $n$ short reads with length $m$, the total size of all the k-mers is equal to $k(m-k+1)n$. We have
\begin{equation}
k(m-k+1)= \frac{(m+1)^2}{4}-(\frac{m+1}{2}-k)^2. \nonumber
\end{equation}
Hence, the size is peaked when $k=(m+1)/2$.  The second effect is the compression ratio of MSP for larger $k$ is higher.  These two figures demonstrate the second effect dominates, since we do not observe a peak at $k=(m+1)/2$.  The result is also in line with the analytical conclusion made for the random string model (see Theorems \ref{lp1m} and \ref{thrm:peakmemory}).

\section{Related Work}
High throughput sequencing technologies are generating tremendous amounts of short reads data.  Assembling these datasets becomes a critical research topic. With the development of next-generation sequencing techniques, the de Bruijn graph sequence assembly approaches became popular, including Euler\cite{pevzner2001eulerian}, Velvet\cite{zerbino2008velvet}, AllPaths\cite{butler2008allpaths}, SOAPdenovo\cite{li2010novo}, etc.

All these de Bruijn graph based algorithms have to solve a critical problem in the process of constructing de Bruijn graph, which merges duplicate k-mers into the same vertex.  When the number of short reads comes to the level of billions, the de Bruijn graph can easily consume hundreds of gigabytes of memory. Several algorithms have been proposed to solve the memory overwhelming problem of graph-based assemblers. Simpson and Durbin \cite{simpson2010efficient} adopted FM-index \cite{ferragina2005indexing} to achieve compression in building the string graph \cite{myers2005fragment}, which is an alternative graph formulation used in sequence assembly (string graph is much more expensive to construct than de Bruijn graph, so it is not as popular as de Bruijn graph). However, the step of building the suffix array and FM-index is very time-consuming and memory-intensive. Orthogonally, Conway and Bromage \cite{conway2011succinct} used succinct bitmap data structure to compress the representation of de Bruijn graph. But the overall space requirement will still increase as the graph becomes ``bigger" (more nodes and edges). Distributed assembly algorithms were also proposed, e.g., ABySS\cite{simpson2009abyss} and Contrail\cite{schatz2010contrail}.  They partition k-mers in a distributed manner to avoid memory bottleneck. Unfortunately, using a hash function to distribute k-mers evenly across a cluster cannot ensure adjacent k-mers being mapped to the same machine.  It results in intense cross-machine communications since adjacent k-mers form edges in the graph. The proposed minimum substring partitioning technique solves this problem: it not only generates small partitions, but also retains adjacent k-mers in the same partition.

The de Bruijn graph construction problem is related to duplicate detection.  The traditional duplicate detection algorithms perform a merge sort to find duplicates, e.g., Bitton and DeWitt \cite{BiDe83}. Teuhola and Wegner \cite{teuhola1991minimal} proposed an $O(1)$ extra space, linear time algorithm to detect and delete duplicates from a dataset.  Teuhola \cite{Teuhola93externalduplicate} introduced an external duplicate deletion algorithm that makes an extensive use of hashing.  It was reported that hash-based approaches are much faster than sort/merge in most cases.  Bucket sort \cite{CLRS01} is adoptable to these techniques, which works by partitioning an array into a number of buckets. Each bucket is then sorted individually.  By replacing sort with hashing, it can solve the duplicate detection problem too.  Duplicate detection has also been examined in different contexts, e.g., stream \cite{MDA05} and text \cite{BiMo'03}.  A survey for general duplicate record detection solutions was given by Elmagarmid, Ipeirotis and Verykios \cite{AGV07}.

The problem setting of de Bruijn graph construction is different from duplicate detection in sense that elements in short reads are highly overlapped and a de Bruijn graph needs to find which element is a duplicate to which.  The proposed minimum substring partitioning technique can utilize the overlaps to reduce the partition size dramatically.  Meanwhile, the three steps, partitioning, mapping, and merging for disk-based de Bruijn graph construction can efficiently connect duplicate k-mers scattered in different short reads into the same vertex.

The concept of minimum substring was introduced in \cite{Robe+04} for memory-efficient sequence comparison. Our work develops minimum substring based partitioning and its use in sequence assembly.  We also theoretically analyze several important properties of minimum substring partitioning.

\section{Conclusions}
We introduced a new partitioning concept - minimum substring partitioning (MSP), which is appropriate and efficient to solve the duplicate k-mer merging problem in the assembly of massive short read sequences. It makes use of the inherent overlaps among k-mers to generate compact partitions.  This partitioning technique was successfully applied to de Bruijn graph construction with very small memory footprint.  We discussed the relations between the partition size and the minimum substring length and analytically derived the capacity of minimum substrings based on a random string model.\nop{We further extended the use of MSP to other applications such as k-mer counting.} Our MSP-based de Bruijn graph construction algorithm\nop{and k-mer counting algorithm were} was evaluated on real DNA short read sequences. Experimental results showed that it can not only successfully complete the tasks on very large datasets within a small amount of memory, but also achieve better performance than existing state-of-the-art algorithms.


\balance

\nop{
\section{ACKNOWLEDGMENTS}
This research was sponsored in part by the U.S. National Science Foundation under grant IIS-0905084 and by the Army Research Laboratory under cooperative agreement W911NF-09-2-0053 (NS-CTA). X. Yan was supported in part by the Institute for Collaborative Biotechnologies through Grant, No. DFR3A-8-447850-23002 from the U.S. Army Research Office.  The views and conclusions contained herein are those of the authors and should not be interpreted as representing the official policies, either expressed or implied, of the Army Research Laboratory or the U.S. Government. The U.S. Government is authorized to reproduce and distribute reprints for Government purposes notwithstanding any copyright notice herein.
}


\section{Appendix}
We design an efficient polynomial-time algorithm for computing the
probability that a given $p$-substring is the minimum $p$-substring of a random $n$-length string $S$.
The complication arises because of the huge overlaps among the $p$-substrings of $S$:
each $p$-substring shares $p-1$ symbols with its predecessor, so these subproblems
are not \emph{independent}. We design a non-trivial \emph{dynamic programming}
algorithm that circumvents this complication, and leads to an $O(n^2)$ algorithm.
Because the underlying problem is quite general, we find it best to describe the
problem and its solution using the following abstract setting.

Let $S \:=\: s_1 s_2 \ldots s_n$ be a random string (the DNA sequence), where each
letter $s_i$ is an independent random variable taking values from the set
$\Sigma \;=\; \{0, 1, 2, 3 \}$ with probabilities $p_0, p_1, p_2, p_3$, respectively.
That is, $s_i$ assumes value $j$ with probability $p_j$, for $j=0,1,2,3$, and
these probabilities sum to $1$, namely, $\sum_{j=1}^4 p_j = 1.$
We will use the notation $S_i$ for the prefix substring of $S$ of length $i$, namely,
$s_1 s_2 \ldots s_i$, and $S(j)$ for its suffix substring of length $j$,
namely, $s_{n-j+1} \ldots s_n$.
The notation $S_i (j)$ will be used for the $j$ symbol long suffix of the
prefix substring $S_i$ ($j \leq i$), namely, $s_{i-j+1} \ldots s_i$ (see Figure~\ref{fig:string}).
We will adopt the convention that substrings of length zero are empty;
in particular, $S_i(0)$ and $S_0(j)$ are empty strings.
Any two substrings of equal length can be compared using the lexicographical order, and
we will use the standard notation $<, \leq, = , \geq, >$ to denote their relative order.

In order to distinguish the target string $W$ from the DNA sequence, we will
call the former a \emph{word}. In particular, given an $m$-word $W$ (to distinguish the abstract problem from the real problem, here we use {\it m} instead of {\it p}), also on
the alphabet $\Sigma \;=\; \{0, 1, 2, 3 \}$, we wish to compute the probability
that no $m$-substring of $S$ is smaller than or equal to $W$. More specifically, what
is the probability that $S_i(m) \; > \; W$, for all $i=m, m+1, \ldots, n$. As we will argue later, if we know this probability for W and the m-word
immediately preceding W in the lexicographical ordering, then by calculating their
difference we can get the probability that W {\it itself } is the minimum m-substring, which is what we ultimately need.

In order to build some intuition into the problem, let us consider the prefix $S_i$.
Let us call $S_i$ \emph{clean} if it does not contain an $m$-substring $\leq W$.
Suppose we inductively assume $S_{i-1}$ to be clean.
Then, it follows that $S_i$ is clean only if $S_i (m) > W$. In other words, to ensure
that a prefix substring $S_i$ is clean we need two conditions: (1) the substring
$S_{i-1}$ is clean, (2) the \emph{$m$-suffix} of
$S_i$, $S_i(m)$, is larger than $W$. In fact, we will need these conditions to be \emph{recursively} enforced, meaning that we will need $S_i (j) > W_j$, for all $j$.

\begin{figure}[tbh]
\centering
\epsfig{file=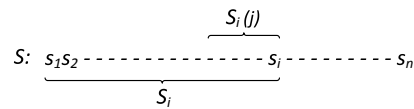, width=0.30\textwidth}
\caption{Illustration of $S$, $S_i$, and the $j$-suffix of $S_i$.}
\vspace{-1mm}
\label{fig:string}
\end{figure}

With this motivation, we now define the 2-dimensional table $Q$, which will form the basis
of our dynamic programming algorithm. The table $Q$ has size $(n+1) \times m$, where
the entry $Q[i,j]$ holds the probability that $S_i$ is clean \emph{and} $S_i(j) > W_j$.
Thus, $Q[i,0]$ is the probability that $S_i$ is clean, and the final value we wish to
compute is $Q[n,0]$, which is the probability that the entire string $S$ is clean, meaning it
has no $m$-substring less than or equal to $W$.
Of course, the probability that \emph{$W$ is the minimum} $m$-word in $S$ is easily
computed as $Q'[n,0] - Q[n,0]$, where $Q'$ is the same dynamic programming table
computed for the target $m$-word $W'$, where $W'$ is the immediate predecessor of $W$
in the lexicographical ordering of $m$-words.

Algorithm {\it MinSTB} (Minimum Substring Tail Bounds) describes in pseudo-code how to compute the $Q$
table in row-major order, with the convention that $Q[0,j]=1$ for all $j$.
Assuming the first $i$ rows of the table have been computed, the algorithm shows how to compute the row $i+1$.
The analysis of the algorithm is given in the following theorem.

\begin{algorithm}
\renewcommand{\thealgorithm}{}
\caption{\hs{\it MinSTB}: Computes the values $Q[i+1,j]$.}
\label{algo:random}
\begin{algorithmic}
\REQUIRE $ 0\leq j\leq m,\; j\leq i\leq n$
\IF {$i+1 < m$}
	\STATE $Q[i\hs+\hs 1,0]\;=\;1$.
	\STATE $Q[i\hs+\hs 1,1]\;=\;\sum_{k > w_1}^3 p_k.$
	\STATE $Q[i\hs+\hs 1,j] \;=\; \sum_{k > w_j}^3 p_k \;\;+\;\; Q[i,j \hs - \hs 1]\cdot p_{w_j},\;\;\; j>1.$
\ELSE
   \STATE $Q[i\hs +\hs 1,0] \;=\;  Q[i,0]\cdot\sum_{k > w_m}^3 p_k
 \;\;+\;\;  Q[i,m\hs - \hs 1]\cdot p_{w_m}$.

   \IF {$w_m > w_j$ and $j>0$}
      \STATE $Q[i\hs +\hs 1,j] \;=\; Q[i,0]\cdot\sum_{k > w_m}^3 p_k
 \;\;+\;\; Q[i,m\hs - \hs 1]\cdot p_{w_m} $.
   \ENDIF
	
	\IF {$w_m < w_j$ and $j>0$}
		\STATE $Q[i\hs +\hs 1,1] \;=\; Q[i,0]\cdot\sum_{k > w_1}^3 p_k$.
		\STATE $Q[i\hs +\hs 1,j] \;=\; Q[i,0]\cdot \sum_{k > w_j}^3 p_k
\;\;+\;\; Q[i,j \hs - \hs 1]\cdot p_{w_j}$ .
   \ENDIF

	\IF {$w_m = w_j$ and $j>0$}
		\STATE $Q[i\hs +\hs 1,1] \;=\; Q[i,0]\cdot \sum_{k > w_1}^3 p_k$.
		\IF {$W_{m-1}(j \hs - \hs 1) >
 W_{j-1}$}
			\STATE $Q[i\hs +\hs 1,j] \;=\; Q[i,0]\cdot\sum_{k > w_j}^3 p_k
 \;\;+\;\;Q[i,m\hs - \hs 1]\cdot  p_{w_j}$.
		\ENDIF
		\IF {$W_{m-1}(j-1)
 \leq W_{j-1}$}
			\STATE $Q[i\hs +\hs 1,j] \;=\; Q[i,0]\cdot \sum_{k > w_j}^3 p_k
 \;\;+\;\;Q[i,j \hs - \hs 1]\cdot p_{w_j}$.
		\ENDIF
   \ENDIF
\ENDIF
\end{algorithmic}
\end{algorithm}

\begin{theorem}
Given a random string $S$ and an $m$-word $W$ on $\Sigma=\{0,1,2,3\}$, we can compute the probability
that $S$ has no $m$-substring $\leq W$ in $O(n^2)$ time.
\end{theorem}

\begin{proof}
We prove how Algorithm {\it MinSTB} correctly computes the $(i\hs + \hs 1)$th row of the table $Q$
from the $i$th row.
First consider the case where $i\hs + \hs 1 < m$. Then $S_{i+1}$ has no $m$-substring, and we just need that $S_{i+1}(j) > W_j$. If $j=0$,
then $Q[i\hs + \hs 1,j] \;=\; 1$. If $j=1$, then we simply need that $s_{i+1} > w_1$. Thus
$$Q[i\hs + \hs 1,1] \;=\; \sum_{k > w_1}^3\hs p_k.$$
Finally if $j > 1$, then all we need is that either (1) $s_{i+1} > w_j$, or
(2) $s_{i+1}=w_j$ and $S_i(j \hs - \hs 1) > W_{j-1}$. Therefore:
$$Q[i\hs + \hs 1,j] \;=\; \sum_{k > w_j}^3\hs p_k \;\;+\;\; Q[i,j\hs-\hs 1]\cdot p_{w_j},$$
where $\sum_{k > w_j}^{3} p_k$ is the probability that $s_{i+1} > w_j$, $p_{w_j}$ is the probability that $s_{i+1}=w_j$, and $Q[i,j \hs - \hs 1]$ is the probability that
$S_i$ is clean and $S_i(j \hs - \hs 1) > W_{j-1}$.

Now consider the case where $i+1 > m$.
Suppose $S_i$ is clean, then $S_{i+1}$ is not clean if and only if $S_{i+1}(m)\leq W$.
This will {\bf not} happen if and only if $s_{i+1} > w_m$,
or $s_{i+1}=w_m$ but $S_i(m\hs - \hs 1) > W_{m-1}$. Therefore
$$
Q[i\hs + \hs 1,0] \;=\; Q[i,0]\cdot\hs\sum_{k > w_m}^{3}\hs p_k  \;\;+\;\; Q[i,m\hs - \hs 1]\cdot p_{w_m},
$$
where $Q[i,0]$ is the probability that $S_i$ is clean.

The computation of $Q[i\hs + \hs 1,j]$ for $j\neq 0$ depends on
the value of $w_m$ compared to $w_j$. This stems from the fact
that if $w_m < w_j$, then we only need to compare $S_{i+1}$ against $W_j$, i.e., if
$S_i(j \hs - \hs 1) > W_{j-1}$ then necessarily $S_i(m\hs - \hs 1) > W_{m-1}$.
In fact, there are three cases:

\begin{enumerate}[leftmargin=11pt]

\item $w_m > w_j$. In this case if $s_{i+1} > w_m >w_j$ then $S$ cannot have any $m$-substring
$\leq W_m$, or any $j$-substring $\leq W_j$, which ends at index $i$.
So all we need is for $S_i$ to be clean.
If $s_{i+1} < w_m$ then $S_{i+1}$ is not clean. If $s_{i+1} = w_m$, then $s_{i+1} > w_j$ and $S_{i+1}(j) > W_j$. In this
case we just need that $S_i(m\hs - \hs 1) > W_{m-1}$, and we have
$$Q[i\hs + \hs 1,j] \;=\; Q[i,0]\cdot\hs\sum_{k > w_m} p_k  \;\;+\;\; Q[i,m\hs - \hs 1]\cdot p_{w_m}.$$
This holds if $j=1$, since $s_{i+1}>w_j$ even if $s_{i+1}=w_m$.

\item $w_m < w_j$. The argument is similar to the previous case: if $s_{i+1} > w_j > w_m$, then all we need
is for $S_i$ to be clean. If $s_{i+1} = w_j$ then if $j=1$,
$$Q[i\hs + \hs 1,1] \;=\;Q[i,0]\cdot\hs\sum_{k > w_1} p_k,$$
but if $j>1$ we need that $S_i(j \hs - \hs 1) > W_{j-1}$.
Therefore
$$Q[i\hs + \hs 1,j] \;=\;Q[i,0]\cdot\hs\sum_{k > w_j} p_k  \;\;+\;\;  Q[i,j \hs - \hs 1]\cdot p_{w_j}.$$

\item $w_m = w_j$. If $s_{i+1} > w_m$ then all we need is for $S_i$ to be clean. If $s_{i+1} = w_m = w_j$,
then if $j=1$
$$Q[i\hs + \hs 1,1] \;=\;Q[i,0]\cdot\hs\sum_{k > w_j} p_k.$$
If $j>1$ there are two cases:

\begin{itemize}[leftmargin=9pt]
\item $W_{m-1}(j \hs - \hs 1) > W_{j-1}$. Then if $S_{i+1}$ is clean,
necessarily $S_{i+1}(j) > W_j$ because $s_{i+1}=w_j$
and $S_{i+1}(m\hs - \hs 1) > W_{m-1} $, which implies that
$S_{i+1}(j \hs - \hs 1) > W_{j-1}$. Therefore we only need $S_{i+1}$ to be
clean, which happens only if $S_i(m\hs - \hs 1) > W_{m-1}$.
In this case:
$$Q[i\hs + \hs 1,j] \;=\;Q[i,0]\cdot\hs\sum_{k > w_j} p_k   \;\;+\;\; Q[i,m\hs - \hs 1]\cdot p_{w_j} .$$
\vspace{-0.3cm}
\item $W_{m-1}(j \hs - \hs 1) \leq W_{j-1}$. The argument is the same as the previous case, except that now
we need $S_i(j \hs - \hs 1) > W_{j-1}$, which happens with probability $Q[i,j \hs - \hs 1]$. Thus
$$Q[i\hs + \hs 1,j] \;=\; Q[i,0]\cdot\hs\sum_{k > w_j} p_k \;\;+\;\; Q[i,j \hs - \hs 1]\cdot p_{w_j}.$$
\vspace{-0.4cm}
\end{itemize}

\end{enumerate}

Each entry of the table can be computed in constant time, and therefore the whole table can be
computed in $O(n^2)$.
\end{proof}

\end{document}